\documentclass[aps,pre,preprint,groupedaddress]{revtex4-1}

\usepackage{verbatim}

\usepackage{color}

\usepackage{amsthm}
\usepackage{amsmath}
\usepackage{amsfonts}
\usepackage{amssymb}

\usepackage{graphicx}
\graphicspath{ {fig/} }

\newcommand{\real}{\mathbb{R}}

\newcommand{\parenth}[1]{ \left( #1 \right) }

\newtheorem{theorem}{Theorem}[section]
\newtheorem{lemma}[theorem]{Lemma}

\begin{document}


\title{Gradient Navigation Model for Pedestrian Dynamics}


\author{Felix Dietrich}
\email{felix.dietrich@tum.de}
\affiliation{Department of Computer Science and Mathematics, Munich University of Applied Sciences, 80335 Munich, Germany}
\affiliation{Zentrum Mathematik, Technische Universit\"{a}t M\"{u}nchen, 85748 Garching, Germany}

\author{Gerta K\"{o}ster}
\email{gerta.koester@hm.edu}
\affiliation{Department of Computer Science and Mathematics, Munich University of Applied Sciences, 80335 Munich, Germany}

\date{\today}

\begin{abstract}
We present a new microscopic ODE-based model for pedestrian dynamics: the Gradient Navigation Model. 
The model uses a superposition of gradients of distance functions to directly change the direction of the velocity vector. The velocity is then integrated to obtain the location. The approach differs fundamentally from force based models needing only three equations to derive the ODE system, as opposed to four in, e.g., the Social Force Model. 
Also, as a result, pedestrians are no longer subject to inertia. Several other advantages ensue: Model induced oscillations are avoided completely since no actual forces are present. The derivatives in the equations of motion are smooth and therefore allow the use of fast and accurate high order numerical integrators. At the same time, existence and uniqueness of the solution to the ODE system follow almost directly from the smoothness properties. 
In addition, we introduce a method to calibrate parameters by theoretical arguments based on empirically validated assumptions rather than by numerical tests. These parameters, combined with the accurate integration, yield simulation results with no collisions of pedestrians. 
Several empirically observed system phenomena emerge without the need to recalibrate the parameter set for each scenario: obstacle avoidance, lane formation, stop-and-go waves and congestion at bottlenecks. The density evolution in the latter is shown to be quantitatively close to controlled experiments. Likewise, we observe a dependence of the crowd velocity on the local density that compares well with benchmark fundamental diagrams.
\end{abstract}


\maketitle

\section{Introduction}


Pedestrian flows are dynamical systems. Numerous models exist \citep{hamacher-2001,antonini-2006,chraibi-2011} on both the macroscopic \citep{hughes-2001,hoogendoorn-2004} and the microscopic level \citep{helbing-1995,chraibi-2010,seitz-2012}. In the latter, two approaches seem to dominate: ordinary differential equation (ODE) models and cellular automata (CA).

ODE models are particularly well suited to describe dynamical systems because they can formally and concisely describe the change of a system over time. The mathematical theory for ODEs is rich, both on the analytic and the numerical side.
In CA models, pedestrians are confined to the cells of a grid. They move from cell to cell according to certain rules. This is computationally efficient, but there is only little theory available \citep{boccara-2003}.
Many CA models employ a floor field to steer individuals around obstacles \cite{burstedde-2001,ezaki-2012}.

The use of floor fields for pedestrian navigation in ODE models is only sparingly described in literature. In \citep{helbing-1995}, pedestrians steer towards the edges of a polygonal path, in \citep{hoogendoorn-2003} optimal control is applied.
In addition, most of the ODE models are derived from molecular dynamics where the direction of motion is gradually changed by the application of a force. This leads to various problems mostly caused by inertia \citep{chraibi-2011, koster-2013}.

Cellular automata and more recent models in continuous space, like the Optimal Steps Model \citep{seitz-2012}, deviate from this approach and directly modify the direction of motion. This is also true for some ODE models in robotics, where movement is very controlled and precise and thus inertia is negligible \citep{starke-2002,starke-2011}.
The direct change of the velocity constitutes a strong deviation from molecular dynamics and hence from force based models.

This paper proposes an application of this model type to pedestrian dynamics: the Gradient Navigation Model (GNM). The GNM is a system of ODEs that describe movement and navigation of pedestrians on a microscopic level. Similar to CA models, pedestrians steer directly towards the direction of steepest decent on a given navigation function. This function combines a floor field and local information like the presence of obstacles and other pedestrians in the vicinity.

The paper is structured as follows: In the model section, three main assumptions about pedestrian dynamics are stated. They lead to a system of differential equations. A brief mathematical analysis of the model is given in the next section where we use a plausibility argument to reduce the number of free parameters in the ODE system from four to two. We constructed the model functions so that they are smooth. Thus, using standard mathematical arguments, existence and uniqueness of the solution follow directly.
In the simulations section the calibrated model is validated against several scenarios from empirical research: congestion in front of a bottleneck, lane formation, stop-and-go waves and speed-density relations. We also demonstrate computational efficiency using high-order accurate numerical solvers like MATLABs \textsf{ode45} \citep{dormand-1980b} that need the smoothness of the solution to perform correctly. We conclude with a discussion of the results and possible next steps.


\section{\label{sec-GNM}Model}

The Gradient Navigation Model (GNM) is composed of a set of ordinary differential equations to determine the position $x_i\in\real^2$ of each pedestrian $i$ at any point in time. A navigational vector $\vec{N}_i$ is used to describe an individual's direction of motion.
The model is constructed using three main assumptions.

\begin{description}
\item[Assumption 1] Crowd behavior in normal situations is governed by the decisions of the individuals rather than by physical, interpersonal contact forces.
\end{description}
This assumption is based on the observation that even in very dense but otherwise normal situations, people try not to push each other but rather stand and wait.
It enables us to neglect physical forces altogether and focus on personal goals. If needed in the future, this assumption could be weakened and additional physics could be added similar to \cite{hoogendoorn-2003} who split up what they call the \textit{physical model} and the \textit{control model}.
Note that this assumption sets the GNM apart from models for situations of very high densities, where pedestrian flow becomes similar to a fluid \cite{hughes-2001,helbing-2007}.

\begin{description}
\item[Assumption 2] Pedestrians want to reach their respective targets in as little time as possible based on their information about the environment.
\end{description}
Most models for pedestrian motion are designed with this assumption. Differences remain regarding the optimality criteria for `little time' as well as the amount of information each pedestrian possesses. \cite{helbing-1995} uses a polygonal path around obstacles for navigation, \cite{hoogendoorn-2004b} solve a Hamilton-Jacobi equation, incorporating other pedestrians.
In this paper, we use the eikonal equation similar to \cite{hughes-2001, hartmann-2010} to find shortest arrival times $\sigma$ of a wave emanating from the target region. This allows us to compute the part of the direction of motion $\vec{N}_{i,T}$ that minimizes the time to the target:
\begin{equation}\label{eq:NT}
\vec{N}_{i,T}=-\nabla \sigma
\end{equation}

\begin{description}
\item[Assumption 3] Pedestrians alter their desired velocity as a reaction to the presence of surrounding pedestrians and obstacles. They do so after a certain reaction time.
\end{description}
The relation of speed and local density has been studied numerous times and its existence is well accepted. The actual form of this relation, however, differs between cultures, even between different scenarios \cite{seyfried-2005,chattaraj-2009,jelic-2012}.
Note that assumption 3 not only claims the existence of such a relation but makes it part of the thought process. In our model, we implement this by modifying the desired direction of motion with a vector $\vec{N}_{i,P}$ so that pedestrians keep a certain distance from each other and from obstacles. In models using velocity obstacles, this issue is addressed further \cite{fiorini-1998,shiller-2001,berg-2011,curtis-2014}.
Attractants like windows of a store or street performers could also be modelled as proposed by \cite[p. 49]{molnar-1996}, but are not considered in this paper.
\begin{equation}\label{eq:NP}
\vec{N}_{i,P}=-(\underbrace{\sum_{j\neq i} \nabla P_{i,j}}_\text{influence of pedestrians} + \underbrace{\sum_B \nabla P_{i,B}}_\text{influence of obstacles})
\end{equation}
$\nabla P_{i,j}$ and $\nabla P_{i,B}$ are gradients of functions that are based on the distance to another pedestrian $j$ and obstacle $B$ respectively. Their norm decreases monotonically with increasing distance. To model this, we introduce a smooth exponential function with compact support $R>0$ and maximum $p/e>0$ (see Fig. \ref{fig:calibrate_potentials}):
\begin{equation}\label{eq:h}
h(r;R,p)=\begin{cases}
p\ \text{exp}\parenth{\frac{1}{(r/R)^2-1}}&|r/R|<1\\
0&\text{otherwise}
\end{cases}
\end{equation}
\begin{figure}
\includegraphics[width=0.4\textwidth]{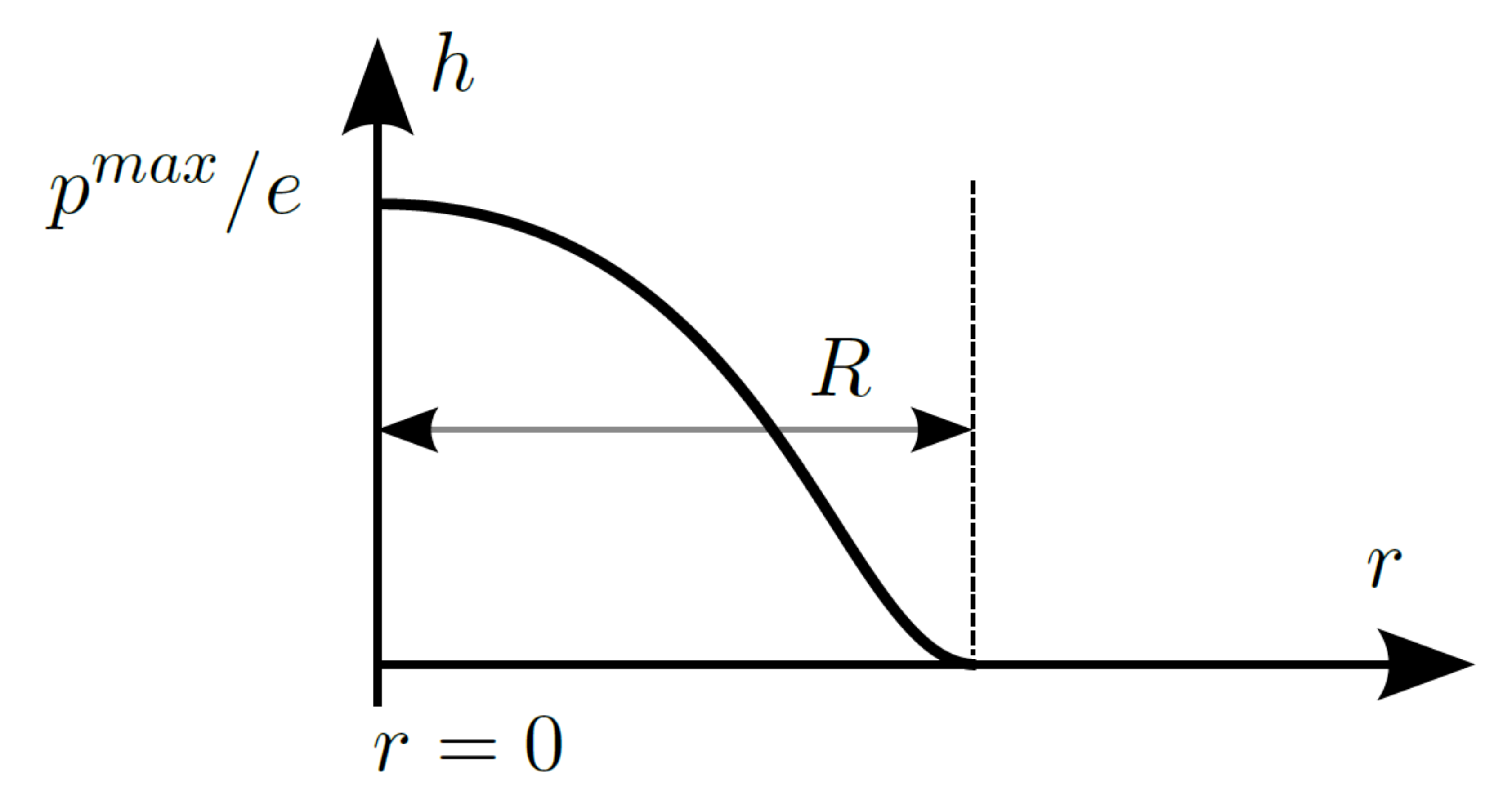}
\caption{\label{fig:calibrate_potentials}The graph of $h$, which depends on the distance between pedestrians $r$ as well the maximal value $p/e$ and support $R$ of $h$.}
\end{figure}
To take the viewing angle of pedestrians into account, we scale $\nabla P_{i,j}$ by
\begin{equation}
s_{i,j}=\tilde{g}(\text{cos}(\kappa\phi_{i,\sigma}-\kappa\phi_j))
\end{equation}
The function $\tilde{g}$ is a shifted logistic function (see appendix \ref{eq:tildeg}) and $(\phi_{i,\sigma}-\phi_j)$ is the angle between the direction $\vec{N}_{i,T}$ and the vector from $x_i$ to $x_j$ (see Fig. \ref{fig:viewingdirection}). $\kappa$ is a positive constant to set the angle of view to $\approx 200\deg$.
\begin{figure}
\includegraphics[width=0.4\textwidth]{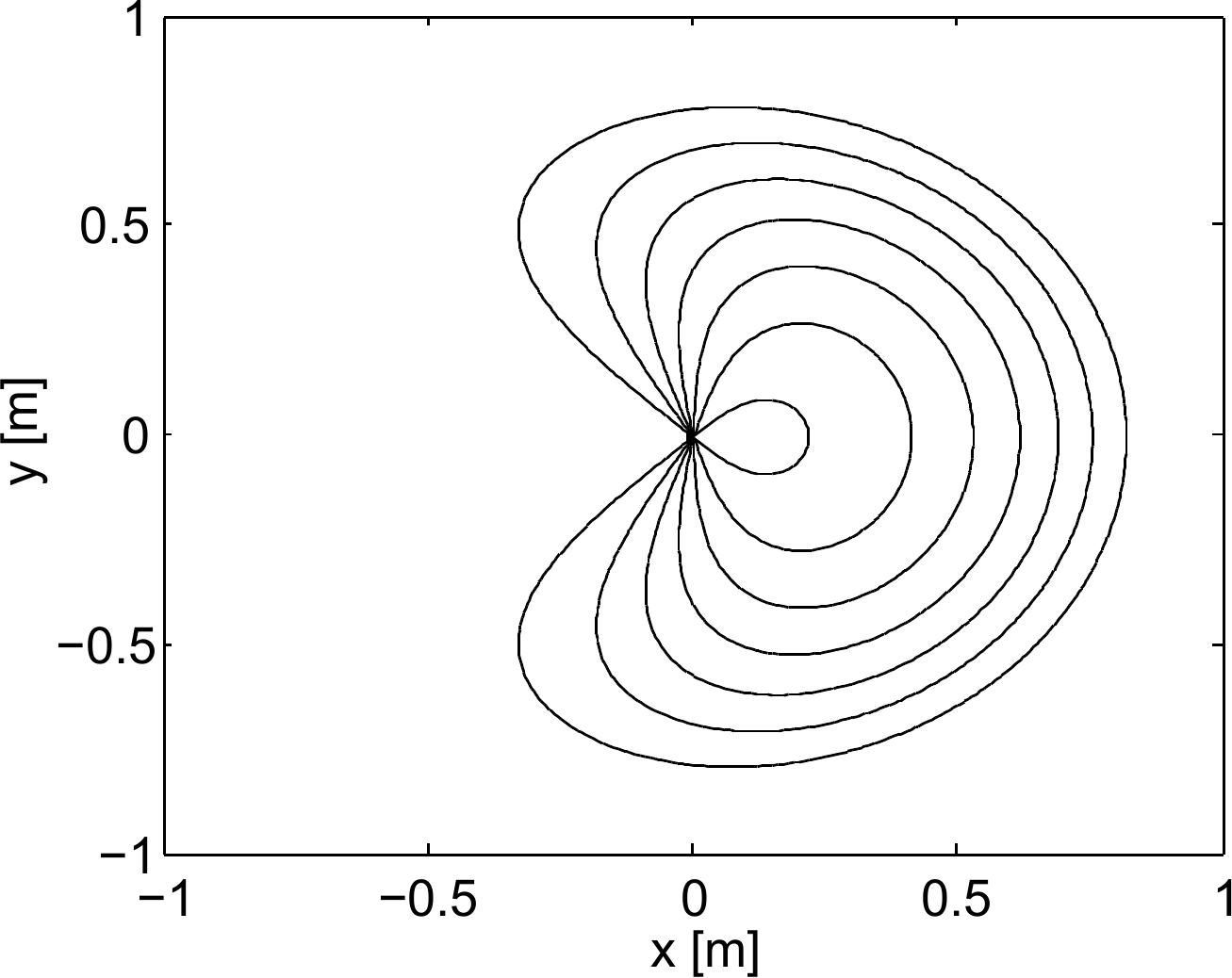}
\caption{\label{fig:viewingdirection}Isolines of the function $s_{i,j}h(\|x_i-x_j\|;1,1)$ with $x_i=0$ and $x_j\in\real^2$. This function represents the field of view of a pedestrian in the origin together with his or her comfort zone. If $x_j$ is close and in front of $x_i$, the function values are maximal, meaning least comfort for $x_i$.}
\end{figure}
Using $h$ and $s_{i,j}$ (see Fig. \ref{fig:viewingdirection}), the gradients in Eq. \ref{eq:NP} are now defined by
\begin{eqnarray}\label{eq:gradient_deltaP}
\nabla P_{i,j}&=&h_\epsilon(\|x_i-x_j\|;p_j,R_j)s_{i,j}\frac{x_j-x_i}{\|x_j-x_i\|}\\\label{eq:gradient_deltaO}
\nabla P_{i,B}&=&h_\epsilon(\|x_i-x_B\|;p_B,R_B)\frac{x_B-x_i}{\|x_B-x_i\|}
\end{eqnarray}
where $p_j,R_j,p_B,R_B$ are positive constants that represent comfort zones between pedestrian $i$, pedestrian $j$ and obstacle $B$. To avoid (mostly artificially induced) situations when pedestrians stand exactly on top of each other \cite{koster-2013}, we replace $h$ by $h_\epsilon$:
\begin{equation}
h_\epsilon(\|x_i-x_j\|;p,R)=h(\|x_i-x_j\|;p,R)-h(\|x_i-x_j\|;p,\epsilon)
\end{equation} where $\epsilon>0$ is a small constant. For $\epsilon\to 0$, $h_\epsilon(\cdot;p,R)\to h(\cdot;p,R)$. For $\|x_i-x_j\|=0$, we also define $\nabla P_{i,j}=0$ and $\nabla P_{i,B}=0$.

To validate the second part of assumption 3, we use the result of \cite{moussaid-2009b}: pedestrians undergo a certain reaction time $\tau$ between an event that changes their motivation and their action. The relaxed speed adaptation is modeled by a multiplicative, time dependent, scalar variable $w:\real^+_0\to\real$, which we call relaxed speed. Its derivative with respect to time, $\dot{w}$, is similar to acceleration in one dimension.

Eq. \ref{eq:NT} and \ref{eq:NP} enable us to construct a relation between the desired direction $\vec{N}$ of a pedestrian and the underlying floor field as well as other pedestrians:
\begin{equation}\label{eq:naviation_vector}
\vec{N}=g(g(\vec{N}_T)+g(\vec{N}_P))
\end{equation}
The function $g:\real^2\to\real^2$ scales the length of a given vector to lie in the interval $[0,1]$. For the exact formula see appendix A. Note that with definition (\ref{eq:naviation_vector}), $N$ must not always have length one, but can also be shorter. This enables us to scale it with the desired speed of a pedestrian to get the velocity vector:
\begin{equation}
\dot{\vec{x}}=\vec{N} w
\end{equation}

With initial conditions $\vec{x}_0=\vec{x}(0)$ and $w_0=w(0)$ the Gradient Navigation Model is given by the equations of motion for every pedestrian $i$:
\begin{equation}\label{eq:GNMequations}
\begin{array}{rcl}
\dot{\vec{x}}_i(t)&=& w_i(t)\vec{N}_i(\vec{x}_i,t)\\\
\dot{w}_i(t)&=&\frac{1}{\tau}\parenth{v_i(\rho(\vec{x}_i))\|\vec{N}_i(\vec{x}_i,t)\|-w_i(t)}\\
\end{array}
\end{equation}

The position $\vec{x}_i:\real\to\real^2$ and the one-dimensional relaxed speed $w_i:\real\to\real$ are functions of time $t$.
$v_i(\rho(\vec{x}_i))$ represents the individuals' desired speeds that depends on the local crowd density $\rho(\vec{x}_i)$ (see assumption 3).  Since the reason for the relation between velocity and density is still an open question \cite{seyfried-2005,jelic-2012}, we choose a very simple relation in this paper: we use  $v_i(\rho)$ constant and normally distributed with mean $1.34$ and standard deviation $0.26$, i.e. $v_i(\rho)=v_i^\text{des}\sim N(1.34,0.26)$. The choice of this distribution is based on a meta-study of several experiments \cite{weidmann-1993}.


With these equations, the direction of pedestrian $i$ changes independently of physical constraints, similar to heuristics in \cite{moussaid-2011}, many CA models and the Optimal Steps Model \cite{seitz-2012}. The speed in the desired direction is determined by the norm of the navigation function $\vec{N}_i$ and the relaxed speed $w_i$.

\section{\label{subsec:staticnavigationfield}The navigation field}
Similar to \cite{hughes-2001} and later \cite{hoogendoorn-2003,hartmann-2010}, we use the solution $\sigma:\real^2\to\real$ to the eikonal equation (\ref{eq:eikonal_equation}) to steer pedestrians to their targets. $\sigma$ represents first arrival times (or walking costs) in a given domain $\Omega\subset\real^2$:
\begin{equation}\label{eq:eikonal_equation}
\begin{array}{rcl}
G(x)\|\nabla \sigma(x)\| &=& 1,\ x\in\Omega\\
\sigma(x)&=&0,\ x\in\Gamma\subset\partial\Omega
\end{array}
\end{equation}


$\Gamma\subset\partial\Omega$ is the union of the boundaries of all possible target regions for one pedestrian. Static properties of a geometry (for example rough terrain or an obstacle) can be modelled by modifying the speed function $G:\real^2\to (0,+\infty)$. \cite{hartmann-2014b} include the pedestrian density in $G$. This enables pedestrians to locate congestions and then take a different exit route. \cite{treuille-2006} used the eikonal equation to steer very large virtual crowds.

If $G(x)=1\ \forall x$, $\sigma$ represents the length of the shortest path to the closest target region. This does not take into account that pedestrians can not get arbitrarily close to obstacles. Therefore, we slow down the wave close to obstacles by reducing $G$ in the immediate vicinity of walls. The influence of walls on $\sigma$ is chosen similar to $\|\nabla P_{i,B}\|$, so that pedestrians incorporate the distance to walls into their route.
 
Being a solution to the eikonal equation (\ref{eq:eikonal_equation}), the floor field $\sigma$ is Lipschitz-continuous \cite{evans-1997}. In the given ODE setting, however, it is desirable to smooth $\sigma$ to ensure differentiability and thus existence of the gradient at all points in the geometry. We employ mollification theory \cite{evans-1997} with a mollifier $\eta$ (similar to $h$ in Eq. \ref{eq:h}) on compact support $B(x)$ to get a mollified $\nabla\sigma$, which we call $\nabla\tilde{\sigma}$:
\begin{equation}\label{eq:smooth_nablasigma}
\nabla\tilde{\sigma}(x)=\nabla(\eta * \sigma)(x)=\int_{B(x)} \nabla\eta(y)\sigma(x-y)dy\in C^\infty(\real^2,\real^2)
\end{equation}



\section{\label{sec-ANA}Mathematical Analysis and Calibration}

Existence and uniqueness of a solution to Eq. \ref{eq:GNMequations} follows from the theorem of Picard and Lindeloef when using the method of vanishing viscosity to solve the eikonal equation \cite{evans-1997} and mollification theory to smooth $\nabla\sigma$ (see Eq. \ref{eq:smooth_nablasigma}).

The system of equations in Eq. \ref{eq:GNMequations} contains several parameters.
\cite{moussaid-2009b,johansson-2007} conducted experiments to find the parameters $\tau$ (relaxation constant, $\approx 0.5$) and $\kappa$ (viewing angle, $\approx 200\deg$, which corresponds to a value of $\kappa\approx 0.6$ here). The following free parameters remain:
\begin{itemize}
\item $p_p$ and $R_p$ define maximum and support of the norm of the pedestrian gradient $\|\nabla P_{i,j}\|$
\item $p_B$ and $R_B$ define maximum and support of the norm of the obstacle gradient $\|\nabla P_{i,B}\|$
\end{itemize}
We use an additional assumption to find relations between these four free parameters:
\begin{description}
\item[Assumption 4] A pedestrian who is enclosed by four stationary other persons on the one side and by a wall on the other side, and who wants to move parallel to the wall, does not move in any direction (see Fig. \ref{fig:calibrate_model}).
\end{description}
This scenario is very common in pedestian simulations and involves many elements that are explicitly modeled: other pedestrians, walls and a target direction. The setup also includes other scenarios: when the wall is replaced by two other pedestrians, the one in the center also does not move if assumption 4 holds. This is because the vertical movement is canceled out by the symmetry of the scenario.
\begin{figure}
\begingroup%
  \makeatletter%
  \providecommand\rotatebox[2]{#2}%
  \ifx\svgwidth\undefined%
    \setlength{\unitlength}{160bp}%
    \ifx\svgscale\undefined%
      \relax%
    \else%
      \setlength{\unitlength}{\unitlength * \real{\svgscale}}%
    \fi%
  \else%
    \setlength{\unitlength}{\svgwidth}%
  \fi%
  \global\let\svgwidth\undefined%
  \global\let\svgscale\undefined%
  \makeatother%
  \begin{picture}(1,1)%
    \put(0,0){\includegraphics[width=\unitlength]{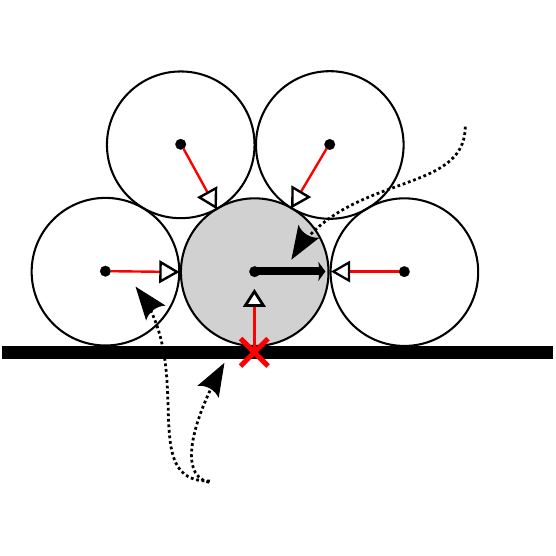}}%
    \put(0.3871858,0.09985346){\color[rgb]{0,0,0}\makebox(0,0)[lb]{\smash{-$\nabla\delta$}}}%
    \put(0.82093761,0.79565877){\color[rgb]{0,0,0}\makebox(0,0)[lb]{\smash{-$\nabla\sigma$}}}%
  \end{picture}%
\endgroup
\caption{\label{fig:calibrate_model}(Color online) Setup used to reduce the number of parameters. The pedestrian in the center (gray) wants to reach a target on the right (black, thick arrow, $-\nabla\sigma$) but is enclosed by four other pedestrians who are not moving and a wall (thick line at the bottom). Together, the pedestrians and the wall act on the gray pedestrian via $-\nabla\delta$ (sum of red, slim arrows). The red cross on the wall marks the position on the wall that is closest to the gray pedestrian.}
\end{figure}
Using assumption 4, we can simplify the system of equations (\ref{eq:GNMequations}) to find dependencies between parameters. First, the direction vectors $\vec{N}_{i,T}$ and $\vec{N}_{i,P}$ are computed based on the given scenario. The gray pedestrian wants to walk parallel to the wall in positive x-direction, that means
\begin{eqnarray}\label{eq:calibrate_sigma}
\vec{N}_{i,T}(x_i)=-\nabla\sigma(x_i)&=&\left[\begin{array}{c}
1\\0
\end{array}\right]
\end{eqnarray}
The remaining function $\vec{N}_{i,P}$ is composed of the repulsive effect of the four enclosing pedestrians and the wall. We simplify equations \ref{eq:gradient_deltaP} and \ref{eq:gradient_deltaO} by taking the limit $\epsilon\to 0$, which is reasonable since the pedestrians do not overlap in the scenario.
\begin{align}\label{eq:calibrate_delta}
\vec{N}_{i,P}=\underbrace{\sum_{i=1}^4 h_p(\|x_\text{gray}-x_i\|)s_{\text{gray},j}\frac{x_\text{gray}-x_j}{\|x_\text{gray}-x_j\|}}_{\text{influence of the `white' pedestrians}}\nonumber\\
+\underbrace{h_B(\|x_\text{gray}-x_B\|)\left[\begin{array}{c}
0\\1
\end{array}\right]}_{\text{influence of the wall}}
\end{align}
Using Eq. (\ref{eq:calibrate_sigma}), (\ref{eq:calibrate_delta}) and assumption 4 the system of differential equations (\ref{eq:GNMequations}) for the pedestrian in the center yield
\begin{eqnarray}
\label{eq:calibration_xdot}\dot{\vec{x}}_i=w \vec{N}_i=-w g(g(\vec{N}_{i,T})+g(\vec{N}_{i,P}))&=&0\\
\dot{w}=\frac{1}{\tau}(v(\rho)\|N\|-w)&=&0
\end{eqnarray}
The second equality yields
\begin{equation}
w=v(\rho)\|\vec{N}\|\implies w\geq 0
\end{equation}
Since assumption 4 does not imply $w=0$, Eq. (\ref{eq:calibration_xdot}) holds true generally if
\begin{equation}\label{eq:calibrate_sigmadelta}
g(\vec{N}_{i,T})=-g(\vec{N}_{i,P})
\end{equation}
Since all $\phi_i$, $x_i$, $x_\text{gray}$ and $x_B$ are known in the given scenario, the only free variables in Eq. (\ref{eq:calibrate_sigmadelta}) are the free parameters of the model: the height and width of $h_p$ (named $p_p$ and $R_p$), as well as $h_B$ (named $p_B$ and $R_B$). With only two equations for four parameters, system (\ref{eq:calibrate_sigmadelta}) is underdetermined and thus we choose $R_B=0.25$ (according to \citep{weidmann-1993}) and $R_p=\sqrt{3}r(\rho_\text{max})$, where $r(\rho_\text{max})$ is the distance of pedestrians in a dense lattice with pedestrian density $\rho_\text{max}$. This choice for $R_p$ ensures that pedestrians adjacent to the enclosing ones have no influence on the one in the center. Note that if this condition is weakened in assumption 4, the model behaves  differently on a macroscopic scale (see Fig. \ref{fig:FD2D_speed_twoneighbors}).

With two of the four parameters fixed, we use Eq. (\ref{eq:calibrate_sigmadelta}) to fix the remaining two. Table \ref{tab:calibrate_parameters} shows numerical values of all parameters, assuming $\rho_{max}=7 P/m^2$, which leads to $r(\rho_{max})\approx 0.41$.
\begin{table}
\begin{tabular}{c|c|l}
Parameter&Value&Description\\
\hline
$\kappa$&0.6&Viewing angle\\
$\tau$&0.5&Relaxation constant\\
\hline
$p_p$&3.59&Height of $h_p$\\
$p_B$&9.96&Height of $h_B$\\
\hline
$R_p$&0.70&Width of $h_p$\\
$R_B$&0.25&Width of $h_B$\\
\end{tabular}
\caption{\label{tab:calibrate_parameters}Numerical values of all parameters of the Gradient Navigation Model using assumption 4 and $\rho_{max}=7P/m^2$. The first two were determined by experiment \cite{johansson-2007,moussaid-2009b}.}
\end{table}


\section{\label{sec-VALIDATION}Simulations}

To solve Eq. (\ref{eq:GNMequations}) numerically, we use the step-size controlling Dormand-Prince-45 integration scheme \cite{dormand-1980b} with $\text{tol}_\text{abs}=10^{-5}$ and $\text{tol}_\text{rel}=10^{-4}$. Employing this scheme is possible because the derivatives are designed to depend smoothly on $x$, $w$ and $t$. Unless otherwise stated, all simulations use the parameters given in Tab. \ref{tab:calibrate_parameters}. The desired speeds $v_i^\text{des}$ are normally distributed with mean $1.34ms^{-1}$ and standard deviation $0.26ms^{-1}$ as observed in experiments \cite{weidmann-1993}. $v_i^\text{des}$ is cut off at $0.3ms^{-1}$ and $3.0ms^{-1}$ to avoid negative or unreasonably high values.
We used the fast marching method \cite{sethian-1999} to solve the eikonal equation (Eq. \ref{eq:eikonal_equation}). The mollification of $\nabla\sigma$ (Eq. \ref{eq:smooth_nablasigma}) is computed using Gauss-Legendre quadrature with $21\times21$ grid points.
All simulations were conducted on a machine with an Intel Xeon (R) X5672 Processor, 3.20 Ghz and with the Java-based simulator VADERE. Simulations of scenarios with over 1000 pedestrians were possible in real time under these conditions.

We validate the model quantitiatively by comparing the flow rates of 180 simulated pedestrians in a bottleneck scenario (see Fig. \ref{fig:liddle_cone}) of different widths with experimentally determined data from \cite{kretz-2006,seyfried-2009,liddle-2011}. The length of the bottleneck is $4m$ in all runs.
Fig. \ref{fig:flow_liddle} shows that, regarding flow rates, the simulation is in good quantitative agreement with data from \cite{kretz-2006,seyfried-2009,liddle-2011} for all bottleneck widths.
\begin{figure}
\includegraphics[width=0.35\textwidth]{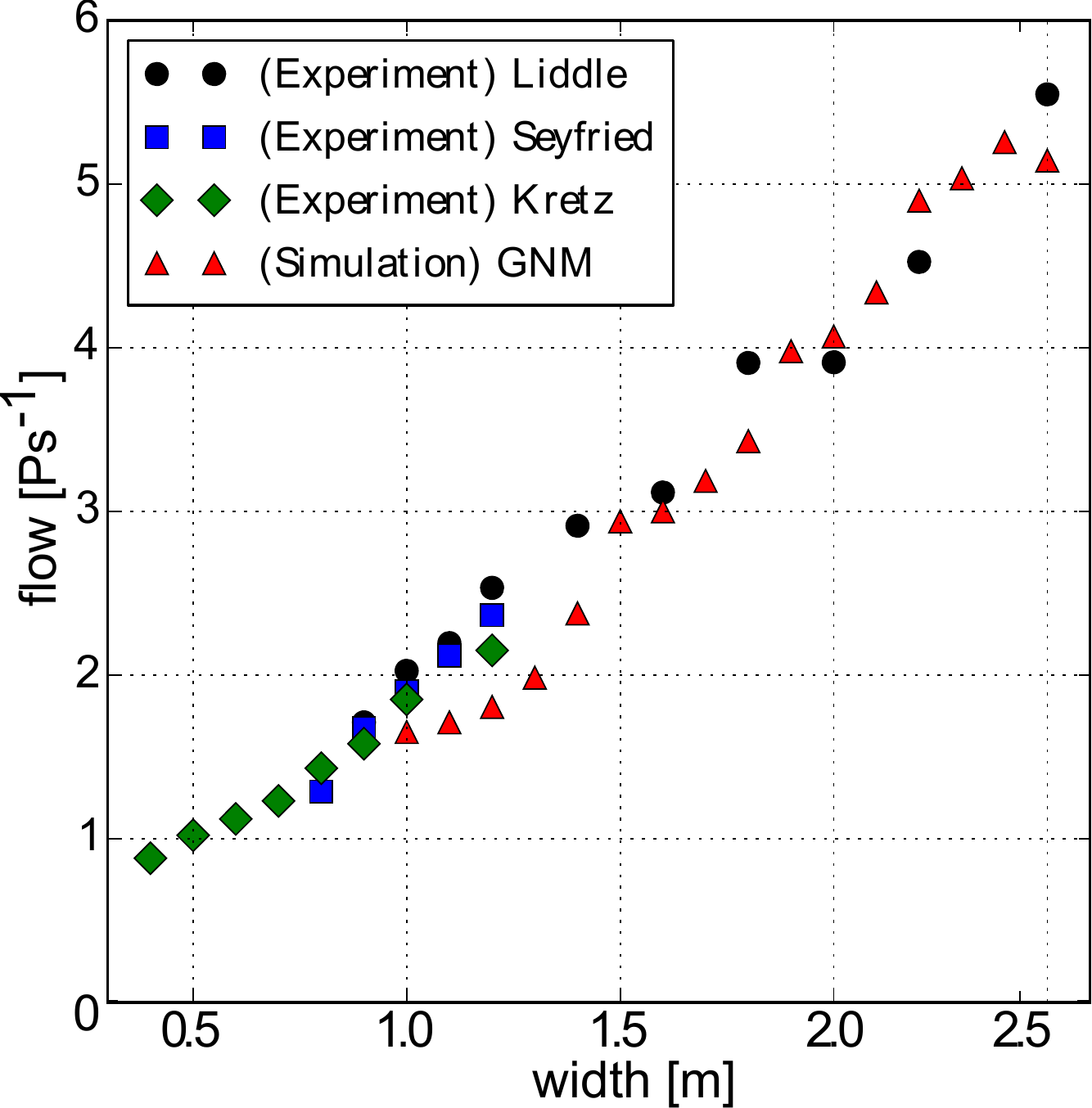}
\caption{\label{fig:flow_liddle}(Color online) Flow rate of the GNM compared to experiments of Kretz \cite{kretz-2006}, Seyfried \cite{seyfried-2009} and Liddle \cite{liddle-2011}. We use the parameters from Tab. \ref{tab:calibrate_parameters} and the normal distribution $N(1.34ms^{-1},0.26ms^{-1})$ to find desired velocities as proposed by \cite{weidmann-1993}.}
\end{figure}
Also, the formation of a crowd in front of a bottleneck matches observations well (see Fig. \ref{fig:liddle_cone}): in front of the bottleneck, they form a cone as observed by \cite{kretz-2006,seyfried-2009b,schadschneider-2011b}. Note that this is different from the behaviour described in \cite{helbing-2000} that tries to capture the dynamics in stress situations. Our simulations suggest that the desired velocity is the most important parameter for this experiment: when we change its distribution to $N(1.57ms^{-1},0.15ms^{-1})$ as found by \cite{gerhardt-2011}, the flow is $\approx 1s^{-1}$ higher for small widths and $\approx 1s^{-1}$ lower for larger widths.
\begin{figure}
\includegraphics[width=0.4\textwidth]{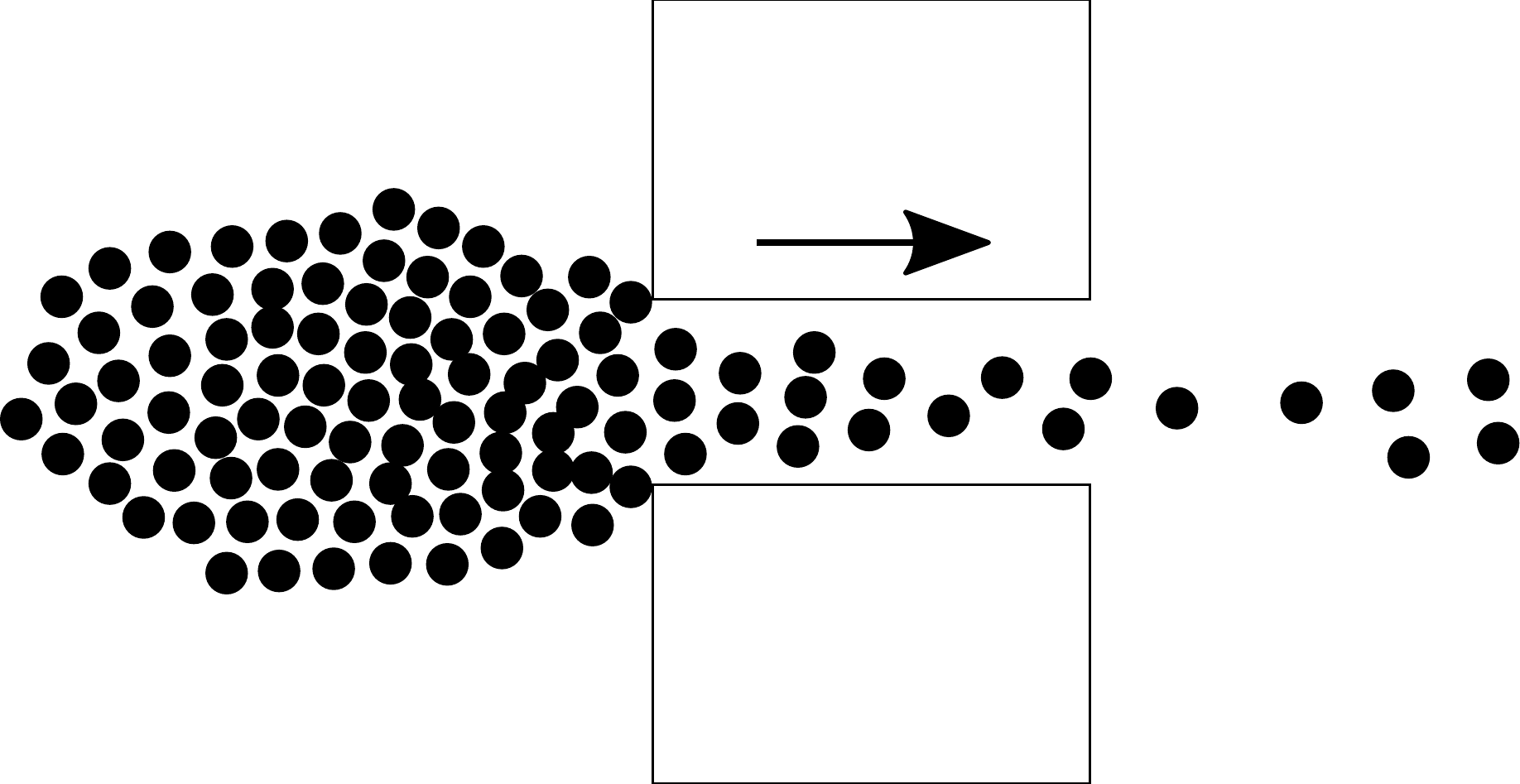}
\caption{\label{fig:liddle_cone}The pedestrians in the GNM simulation form a cone in front of the bottleneck as observed by \cite{kretz-2006,seyfried-2009b,schadschneider-2011b}.}
\end{figure}

The GNM can be calibrated to match the relation of speed and density in a given fundamental diagram. Fig. \ref{fig:FD2D_speed_oneneighbor} shows that for the calibration with only one layer of neighbors, pedestrians do not slow down with increasing densities as quickly as suggested in \cite{weidmann-1993}. When calibrating with one additional layer of pedestrians in the scenario shown in Fig. \ref{fig:calibrate_model}, the curves match much better (see Fig. \ref{fig:FD2D_speed_twoneighbors}). We use the method introduced by \cite{liddle-2011} to measure local density.
\begin{figure}
\includegraphics[width=0.4\textwidth]{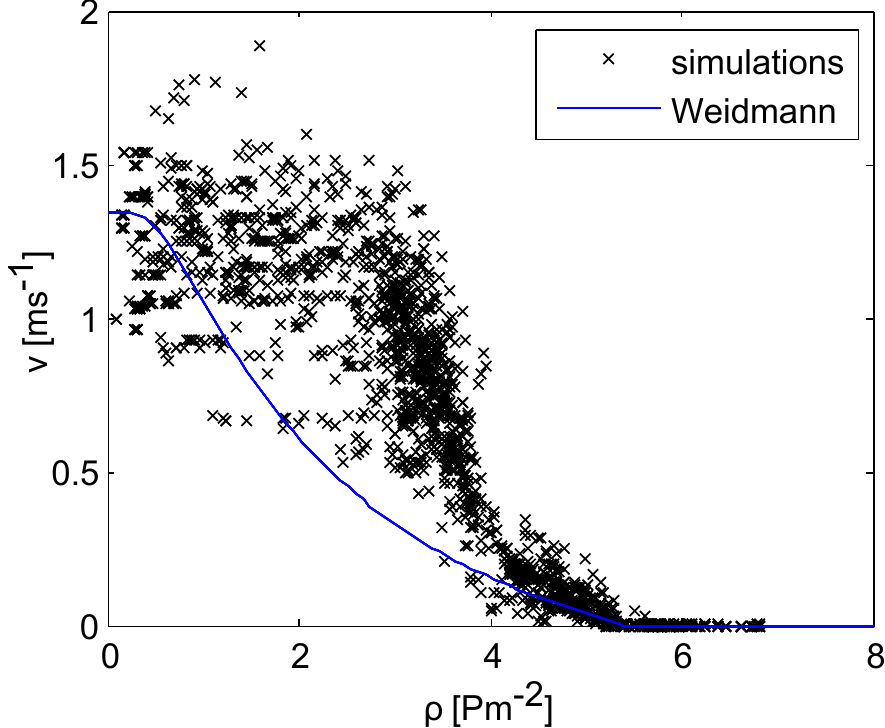}
\caption{\label{fig:FD2D_speed_oneneighbor}(Color online) Speed-density relation in unidirectional flow compared to experimental data from metastudy (Weidmann, \cite{weidmann-1993}). The corridor was $40m$ long and $4m$ wide with periodic boundary conditions. Each cross (labeled `simulation') represents a local measurement at the position of a pedestrian. We use the method introduced by \cite{liddle-2011} to measure local density. The parameter set in these simulations was fixed with the procedure shown in Fig. \ref{fig:calibrate_model} and thus incorporates one layer of four pedestrians.}
\end{figure}
\begin{figure}
\includegraphics[width=0.4\textwidth]{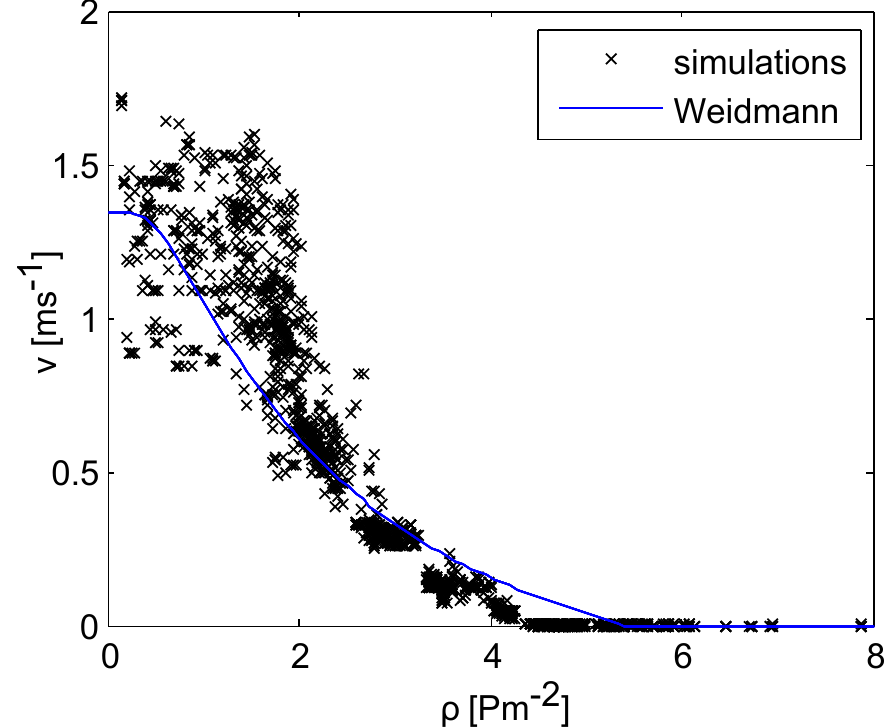}
\caption{\label{fig:FD2D_speed_twoneighbors}(Color online) Speed-density relation in unidirectional flow compared to experimental data from metastudy (Weidmann, \cite{weidmann-1993}). The corridor was $40m$ long and $4m$ wide with periodic boundary conditions. Each cross (labeled `simulation') representsa local measurement at the position of a pedestrian. We use the method introduced by \cite{liddle-2011} to measure local density. The parameter set in these simulations was adjusted with a similar procedure as in Fig. \ref{fig:calibrate_model} to incorporate neighbors of neighbors in the computation of $\nabla\delta$: $R_p=1.0$, $p_p=1.79$, $R_B=0.25$ and $p_B=11.3$.}
\end{figure}

\cite{gaididei-2013,marschler-2013} compute the deviation of distances between drivers to analyze stop-and-go waves in car traffic. No deviation implies no stop-and-go waves since all distances are equal. A large deviation hints at the existence of a wave since there must be regions with large and regions with small distances between drivers. For pedestrian dynamics \cite{helbing-2007, portz-2011, jelic-2012} found stop-and-go waves experimentally. Similar to the wave analysis in traffic, we use the deviation of individual speeds to measure stop-and-go waves. Fig. \ref{fig:stopandgo_mu_sigma} and \ref{fig:stopandgo_scenario} shows that the GNM also produces stop-and-go waves when a certain global density is reached.
\begin{figure}
\includegraphics[width=0.4\textwidth]{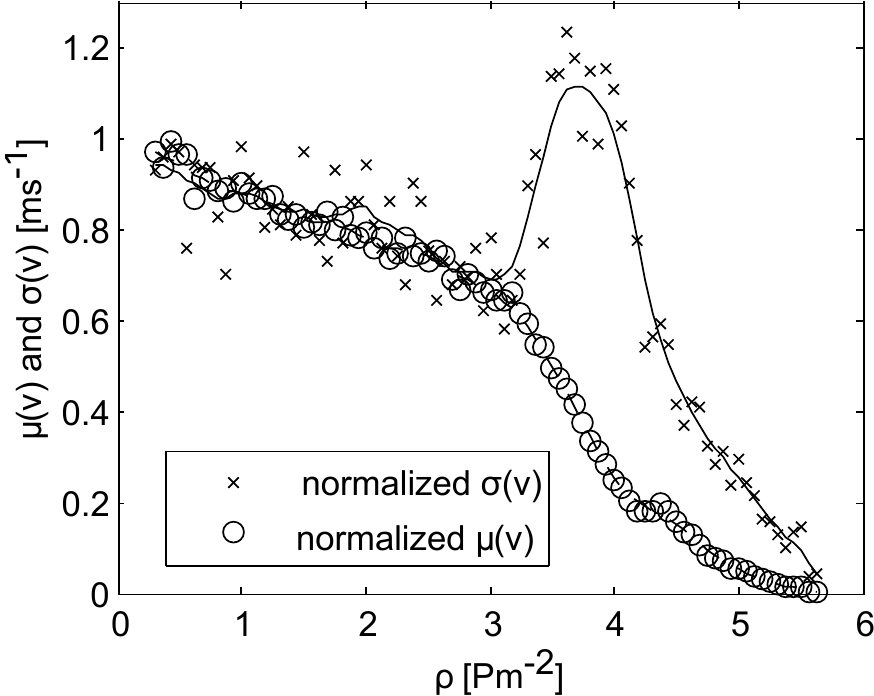}
\caption{\label{fig:stopandgo_mu_sigma}Normalized standard deviation $\sigma(v)/0.26$ (line) and mean $\mu(v)/1.34$ (dashed line) of individual speeds in a unidirectional walkway with differing global densities $\rho$. Both the data points of the simulations and zero-phase digital filtering curves (width: five data points) are shown. The peak of the standard deviation at $\rho=4Pm^{-2}$ indicates stop-and-go waves: even though the mean speed decreases, the speed differences increase, which means that there are regions with low as well as high speeds present at the same time.}
\end{figure}
\begin{figure}
\includegraphics[width=0.4\textwidth]{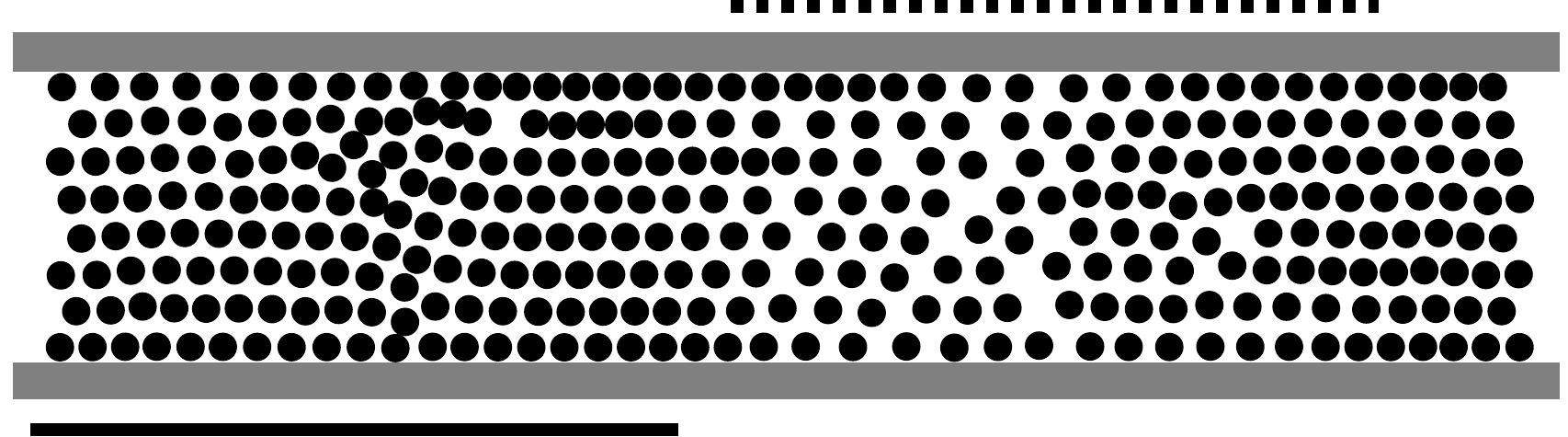}
\caption{\label{fig:stopandgo_scenario}Snapshot of a unidirectional pathway  with global density $\rho=4P/m^2$, dimension $50m\times 4m$, periodic boundary conditions, after $120$ simulated seconds and walking direction to the right. The normal and dashed lines mark slower and faster pedestrians, respectively: a stop-and-go wave.}
\end{figure}

The model also captures lane formation in bidirectional flow out of uniform initial conditions, as observed experimentally by \cite{zhang-2012b}. In the simulation, pedestrians walk bidirectionally in a 10m wide and 150m long pathway at a pedestrian density of $0.3Pm^{-1}$. They start on uniformly distributed positions at the left / right side and walk towards a target on the respective other end. Fig. \ref{fig:lanes_25m} shows that several lanes form. Due to the different desired velocities, many of them brake up after some time. When simulating with densities higher than $1Pm^{-2}$ in the whole pathway, pedestrians block each other and all movement stops.

\begin{figure}
\includegraphics[width=0.45\textwidth]{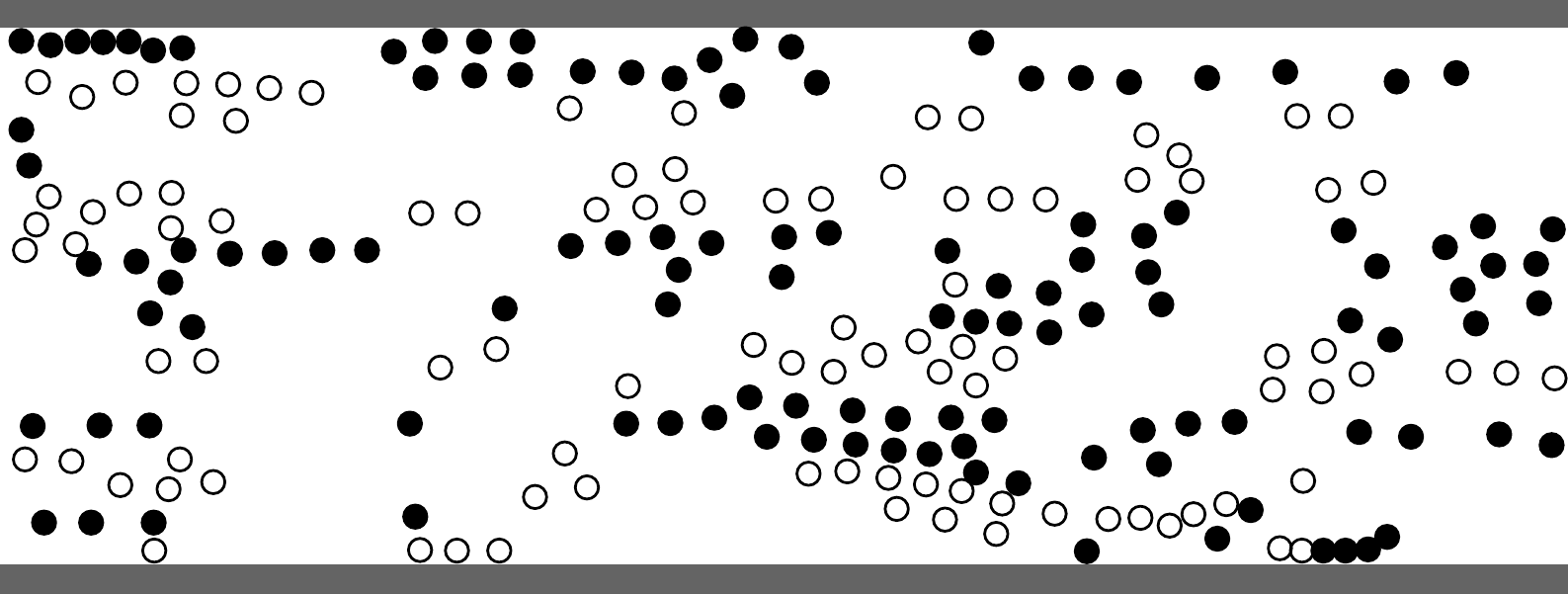}
\caption{\label{fig:lanes_25m}Formation of six lanes in bidirectional flow. Filled circles represent pedestrians walking to the left, empty circles represent pedestrians walking to the right. The walkway is 10m wide and 150m long. The snapshot shows a section of 25m.}
\end{figure}





\section{conclusion}
We introduced a new ODE based microscopic model for pedestrian dynamics, the Gradient Navigation Model. We demonstrated that the model very well reproduces important crowd phenomena, such as  bottleneck scenarios, lane formation, stop-and-go waves and the speed-density relation. In the case of bottlenecks and the speed-density relation good agreement with experimental data was achieved. Calibration of the model parameters was performed using plausible assumptions on the outcome of benchmark scenarios rather than numerical tests. Recalibration for different scenarios was unnecessary.

One main goal for the model was to find a concise formulation with as few equations as possible and, at the same time, certain smoothness properties so that existence, uniqueness and smoothness of the  solution would follow directly. The GNM only needs three equations, as opposed to four in force based models, to describe motion of one pedestrian. In addition, we proposed a floor field to steer pedestrians instead of constructing paths or guiding lines. The floor field was computed by solving the eikonal equation using Sethian's highly efficient fast marching algorithm \cite{sethian-1999}. To achieve smoothness, mollification techniques were employed. The smoothness also enabled us to use numerical schemes of high order making the GNM computationally very efficient.

Two of the methods we introduced can easily be carried over to other models: The plausibility arguments that allowed us to calibrate free parameters hold independently of the model. The mollification techniques that led to the smooth functions could also be used by other differential equation based models like the Social Force Model \cite{helbing-1995, koster-2013} or the Generalized Centrifugal Force Model \cite{chraibi-2010}. 

Some of the most recent enhancements in crowd modeling rely on a floor field to steer pedestrians towards the target. Among them are steering around crowd clusters \cite{hughes-2001, hartmann-2014b} and more sophisticated navigation on the tactical and strategic level \cite{hoogendoorn-2004}. These developments can be employed in the GNM without any change to the equations of motion.

Some empirical observations, such as stop-and-go traffic \cite{helbing-2007, schadschneider-2011b} are not yet well understood, neither from the experimental nor the theoretical point of view. In a mathematical model stability issues and bifurcations often are at the root of such  phenomena. The concise mathematical formulation of the GNM as an ODE systems facilitates stability analysis and the investigation of bifurcations, both tasks that we are currently working on.

\begin{acknowledgments}
This work was partially funded by the German Ministry of Research through the project MEPKA (Grant No. 17PNT028). Support from the TopMath Graduate Center of TUM Graduate School at Technische Universit\"{a}t M\"{u}nchen, Germany, and from the TopMath Program at the Elite Network of Bavaria is gratefully acknowledged. We thank Mohcine Chraibi for his advice on the validation of the model.
\end{acknowledgments}

\bibliography{Literature}

\begin{thebibliography}{46}%
\makeatletter
\providecommand \@ifxundefined [1]{%
 \@ifx{#1\undefined}
}%
\providecommand \@ifnum [1]{%
 \ifnum #1\expandafter \@firstoftwo
 \else \expandafter \@secondoftwo
 \fi
}%
\providecommand \@ifx [1]{%
 \ifx #1\expandafter \@firstoftwo
 \else \expandafter \@secondoftwo
 \fi
}%
\providecommand \natexlab [1]{#1}%
\providecommand \enquote  [1]{``#1''}%
\providecommand \bibnamefont  [1]{#1}%
\providecommand \bibfnamefont [1]{#1}%
\providecommand \citenamefont [1]{#1}%
\providecommand \href@noop [0]{\@secondoftwo}%
\providecommand \href [0]{\begingroup \@sanitize@url \@href}%
\providecommand \@href[1]{\@@startlink{#1}\@@href}%
\providecommand \@@href[1]{\endgroup#1\@@endlink}%
\providecommand \@sanitize@url [0]{\catcode `\\12\catcode `\$12\catcode
  `\&12\catcode `\#12\catcode `\^12\catcode `\_12\catcode `\%12\relax}%
\providecommand \@@startlink[1]{}%
\providecommand \@@endlink[0]{}%
\providecommand \url  [0]{\begingroup\@sanitize@url \@url }%
\providecommand \@url [1]{\endgroup\@href {#1}{\urlprefix }}%
\providecommand \urlprefix  [0]{URL }%
\providecommand \Eprint [0]{\href }%
\providecommand \doibase [0]{http://dx.doi.org/}%
\providecommand \selectlanguage [0]{\@gobble}%
\providecommand \bibinfo  [0]{\@secondoftwo}%
\providecommand \bibfield  [0]{\@secondoftwo}%
\providecommand \translation [1]{[#1]}%
\providecommand \BibitemOpen [0]{}%
\providecommand \bibitemStop [0]{}%
\providecommand \bibitemNoStop [0]{.\EOS\space}%
\providecommand \EOS [0]{\spacefactor3000\relax}%
\providecommand \BibitemShut  [1]{\csname bibitem#1\endcsname}%
\let\auto@bib@innerbib\@empty
\bibitem [{\citenamefont {Hamacher}\ and\ \citenamefont
  {Tjandra}(2001)}]{hamacher-2001}%
  \BibitemOpen
  \bibfield  {author} {\bibinfo {author} {\bibfnamefont {H.~W.}\ \bibnamefont
  {Hamacher}}\ and\ \bibinfo {author} {\bibfnamefont {S.~A.}\ \bibnamefont
  {Tjandra}},\ }\href@noop {} {\emph {\bibinfo {title} {Mathematical Modelling
  of Evacuation Problems: A State of Art}}},\ \bibinfo {type} {Tech. Rep.}\
  (\bibinfo  {institution} {Fraunhofer-Institut f\"{u}r Techno- und
  Wirtschaftsmathematik ITWM},\ \bibinfo {year} {2001})\BibitemShut {NoStop}%
\bibitem [{\citenamefont {Antonini}\ \emph {et~al.}(2006)\citenamefont
  {Antonini}, \citenamefont {Bierlaire},\ and\ \citenamefont
  {Weber}}]{antonini-2006}%
  \BibitemOpen
  \bibfield  {author} {\bibinfo {author} {\bibfnamefont {G.}~\bibnamefont
  {Antonini}}, \bibinfo {author} {\bibfnamefont {M.}~\bibnamefont {Bierlaire}},
  \ and\ \bibinfo {author} {\bibfnamefont {M.}~\bibnamefont {Weber}},\ }\href
  {\doibase 10.1016/j.trb.2005.09.006} {\bibfield  {journal} {\bibinfo
  {journal} {Transportation Research Part B: Methodological}\ }\textbf
  {\bibinfo {volume} {40}},\ \bibinfo {pages} {667} (\bibinfo {year}
  {2006})}\BibitemShut {NoStop}%
\bibitem [{\citenamefont {Chraibi}\ \emph {et~al.}(2011)\citenamefont
  {Chraibi}, \citenamefont {Kemloh}, \citenamefont {Schadschneider},\ and\
  \citenamefont {Seyfried}}]{chraibi-2011}%
  \BibitemOpen
  \bibfield  {author} {\bibinfo {author} {\bibfnamefont {M.}~\bibnamefont
  {Chraibi}}, \bibinfo {author} {\bibfnamefont {U.}~\bibnamefont {Kemloh}},
  \bibinfo {author} {\bibfnamefont {A.}~\bibnamefont {Schadschneider}}, \ and\
  \bibinfo {author} {\bibfnamefont {A.}~\bibnamefont {Seyfried}},\ }\href@noop
  {} {\bibfield  {journal} {\bibinfo  {journal} {Networks and Heterogeneous
  Media}\ }\textbf {\bibinfo {volume} {6}},\ \bibinfo {pages} {425} (\bibinfo
  {year} {2011})}\BibitemShut {NoStop}%
\bibitem [{\citenamefont {Hughes}(2001)}]{hughes-2001}%
  \BibitemOpen
  \bibfield  {author} {\bibinfo {author} {\bibfnamefont {R.~L.}\ \bibnamefont
  {Hughes}},\ }\href {\doibase http://dx.doi.org/10.1016/S0191-2615(01)00015-7}
  {\bibfield  {journal} {\bibinfo  {journal} {Transportation Research Part B:
  Methodological}\ }\textbf {\bibinfo {volume} {36}},\ \bibinfo {pages} {507 }
  (\bibinfo {year} {2001})}\BibitemShut {NoStop}%
\bibitem [{\citenamefont {Hoogendoorn}\ and\ \citenamefont
  {Bovy}(2004{\natexlab{a}})}]{hoogendoorn-2004}%
  \BibitemOpen
  \bibfield  {author} {\bibinfo {author} {\bibfnamefont {S.~P.}\ \bibnamefont
  {Hoogendoorn}}\ and\ \bibinfo {author} {\bibfnamefont {P.~H.~L.}\
  \bibnamefont {Bovy}},\ }\href
  {http://www.sciencedirect.com/science/article/pii/S0191261503000079}
  {\bibfield  {journal} {\bibinfo  {journal} {Transportation Research Part B:
  Methodological}\ }\textbf {\bibinfo {volume} {38}},\ \bibinfo {pages} {169}
  (\bibinfo {year} {2004}{\natexlab{a}})}\BibitemShut {NoStop}%
\bibitem [{\citenamefont {Helbing}\ and\ \citenamefont
  {Moln\'{a}r}(1995)}]{helbing-1995}%
  \BibitemOpen
  \bibfield  {author} {\bibinfo {author} {\bibfnamefont {D.}~\bibnamefont
  {Helbing}}\ and\ \bibinfo {author} {\bibfnamefont {P.}~\bibnamefont
  {Moln\'{a}r}},\ }\href@noop {} {\bibfield  {journal} {\bibinfo  {journal}
  {Physical Review E}\ }\textbf {\bibinfo {volume} {51}},\ \bibinfo {pages}
  {4282} (\bibinfo {year} {1995})}\BibitemShut {NoStop}%
\bibitem [{\citenamefont {Chraibi}\ \emph {et~al.}(2010)\citenamefont
  {Chraibi}, \citenamefont {Seyfried},\ and\ \citenamefont
  {Schadschneider}}]{chraibi-2010}%
  \BibitemOpen
  \bibfield  {author} {\bibinfo {author} {\bibfnamefont {M.}~\bibnamefont
  {Chraibi}}, \bibinfo {author} {\bibfnamefont {A.}~\bibnamefont {Seyfried}}, \
  and\ \bibinfo {author} {\bibfnamefont {A.}~\bibnamefont {Schadschneider}},\
  }\href@noop {} {\bibfield  {journal} {\bibinfo  {journal} {Physical Review
  E}\ }\textbf {\bibinfo {volume} {82}},\ \bibinfo {pages} {046111} (\bibinfo
  {year} {2010})}\BibitemShut {NoStop}%
\bibitem [{\citenamefont {Seitz}\ and\ \citenamefont
  {K\"{o}ster}(2012)}]{seitz-2012}%
  \BibitemOpen
  \bibfield  {author} {\bibinfo {author} {\bibfnamefont {M.~J.}\ \bibnamefont
  {Seitz}}\ and\ \bibinfo {author} {\bibfnamefont {G.}~\bibnamefont
  {K\"{o}ster}},\ }\href {\doibase 10.1103/PhysRevE.86.046108} {\bibfield
  {journal} {\bibinfo  {journal} {Physical Review E}\ }\textbf {\bibinfo
  {volume} {86}},\ \bibinfo {pages} {046108} (\bibinfo {year}
  {2012})}\BibitemShut {NoStop}%
\bibitem [{\citenamefont {Boccara}(2003)}]{boccara-2003}%
  \BibitemOpen
  \bibfield  {author} {\bibinfo {author} {\bibfnamefont {N.}~\bibnamefont
  {Boccara}},\ }\href@noop {} {\emph {\bibinfo {title} {Modeling Complex
  Systems}}},\ \bibinfo {edition} {1st}\ ed.,\ Graduate Texts in Contemporary
  Physics\ (\bibinfo  {publisher} {Springer},\ \bibinfo {year}
  {2003})\BibitemShut {NoStop}%
\bibitem [{\citenamefont {Burstedde}\ \emph {et~al.}(2001)\citenamefont
  {Burstedde}, \citenamefont {Klauck}, \citenamefont {Schadschneider},\ and\
  \citenamefont {Zittartz}}]{burstedde-2001}%
  \BibitemOpen
  \bibfield  {author} {\bibinfo {author} {\bibfnamefont {C.}~\bibnamefont
  {Burstedde}}, \bibinfo {author} {\bibfnamefont {K.}~\bibnamefont {Klauck}},
  \bibinfo {author} {\bibfnamefont {A.}~\bibnamefont {Schadschneider}}, \ and\
  \bibinfo {author} {\bibfnamefont {J.}~\bibnamefont {Zittartz}},\ }\href
  {\doibase 10.1016/S0378-4371(01)00141-8} {\bibfield  {journal} {\bibinfo
  {journal} {Physica A: Statistical Mechanics and its Applications}\ }\textbf
  {\bibinfo {volume} {295}},\ \bibinfo {pages} {507} (\bibinfo {year}
  {2001})}\BibitemShut {NoStop}%
\bibitem [{\citenamefont {Ezaki}\ \emph {et~al.}(2012)\citenamefont {Ezaki},
  \citenamefont {Yanagisawa}, \citenamefont {Ohtsuka},\ and\ \citenamefont
  {Nishinari}}]{ezaki-2012}%
  \BibitemOpen
  \bibfield  {author} {\bibinfo {author} {\bibfnamefont {T.}~\bibnamefont
  {Ezaki}}, \bibinfo {author} {\bibfnamefont {D.}~\bibnamefont {Yanagisawa}},
  \bibinfo {author} {\bibfnamefont {K.}~\bibnamefont {Ohtsuka}}, \ and\
  \bibinfo {author} {\bibfnamefont {K.}~\bibnamefont {Nishinari}},\ }\href
  {\doibase 10.1016/j.physa.2011.07.056} {\bibfield  {journal} {\bibinfo
  {journal} {Physica A: Statistical Mechanics and its Applications}\ }\textbf
  {\bibinfo {volume} {391}},\ \bibinfo {pages} {291} (\bibinfo {year}
  {2012})}\BibitemShut {NoStop}%
\bibitem [{\citenamefont {Hoogendoorn}\ and\ \citenamefont
  {Bovy}(2003)}]{hoogendoorn-2003}%
  \BibitemOpen
  \bibfield  {author} {\bibinfo {author} {\bibfnamefont {S.~P.}\ \bibnamefont
  {Hoogendoorn}}\ and\ \bibinfo {author} {\bibfnamefont {P.~H.~L.}\
  \bibnamefont {Bovy}},\ }\href {\doibase 10.1002/oca.727} {\bibfield
  {journal} {\bibinfo  {journal} {Optimal Control Applications and Methods}\
  }\textbf {\bibinfo {volume} {24}},\ \bibinfo {pages} {153} (\bibinfo {year}
  {2003})}\BibitemShut {NoStop}%
\bibitem [{\citenamefont {K\"{o}ster}\ \emph {et~al.}(2013)\citenamefont
  {K\"{o}ster}, \citenamefont {Treml},\ and\ \citenamefont
  {G\"{o}del}}]{koster-2013}%
  \BibitemOpen
  \bibfield  {author} {\bibinfo {author} {\bibfnamefont {G.}~\bibnamefont
  {K\"{o}ster}}, \bibinfo {author} {\bibfnamefont {F.}~\bibnamefont {Treml}}, \
  and\ \bibinfo {author} {\bibfnamefont {M.}~\bibnamefont {G\"{o}del}},\ }\href
  {\doibase 10.1103/PhysRevE.87.063305} {\bibfield  {journal} {\bibinfo
  {journal} {Physical Review E}\ }\textbf {\bibinfo {volume} {87}},\ \bibinfo
  {pages} {063305} (\bibinfo {year} {2013})}\BibitemShut {NoStop}%
\bibitem [{\citenamefont {Starke}(2002)}]{starke-2002}%
  \BibitemOpen
  \bibfield  {author} {\bibinfo {author} {\bibfnamefont {J.}~\bibnamefont
  {Starke}},\ }in\ \href@noop {} {\emph {\bibinfo {booktitle} {Proceedings of
  the 2002 IEEE International Symposium on Intelligent Control}}}\ (\bibinfo
  {year} {2002})\BibitemShut {NoStop}%
\bibitem [{\citenamefont {Starke}\ \emph {et~al.}(2011)\citenamefont {Starke},
  \citenamefont {Ellsaesser},\ and\ \citenamefont {Fukuda}}]{starke-2011}%
  \BibitemOpen
  \bibfield  {author} {\bibinfo {author} {\bibfnamefont {J.}~\bibnamefont
  {Starke}}, \bibinfo {author} {\bibfnamefont {C.}~\bibnamefont {Ellsaesser}},
  \ and\ \bibinfo {author} {\bibfnamefont {T.}~\bibnamefont {Fukuda}},\ }\href
  {\doibase 10.1016/j.physleta.2011.04.009} {\bibfield  {journal} {\bibinfo
  {journal} {Physics Letters A}\ }\textbf {\bibinfo {volume} {375}},\ \bibinfo
  {pages} {2094} (\bibinfo {year} {2011})}\BibitemShut {NoStop}%
\bibitem [{\citenamefont {Dormand}\ and\ \citenamefont
  {Prince}(1980)}]{dormand-1980b}%
  \BibitemOpen
  \bibfield  {author} {\bibinfo {author} {\bibfnamefont {J.}~\bibnamefont
  {Dormand}}\ and\ \bibinfo {author} {\bibfnamefont {P.}~\bibnamefont
  {Prince}},\ }\href {\doibase 10.1016/0771-050X(80)90013-3} {\bibfield
  {journal} {\bibinfo  {journal} {Journal of Computational and Applied
  Mathematics}\ }\textbf {\bibinfo {volume} {6}},\ \bibinfo {pages} {19}
  (\bibinfo {year} {1980})}\BibitemShut {NoStop}%
\bibitem [{\citenamefont {Helbing}\ \emph {et~al.}(2007)\citenamefont
  {Helbing}, \citenamefont {Johansson},\ and\ \citenamefont
  {Al-Abideen}}]{helbing-2007}%
  \BibitemOpen
  \bibfield  {author} {\bibinfo {author} {\bibfnamefont {D.}~\bibnamefont
  {Helbing}}, \bibinfo {author} {\bibfnamefont {A.}~\bibnamefont {Johansson}},
  \ and\ \bibinfo {author} {\bibfnamefont {H.~Z.}\ \bibnamefont {Al-Abideen}},\
  }\href {\doibase 10.1103/PhysRevE.75.046109} {\bibfield  {journal} {\bibinfo
  {journal} {Physical Review E}\ }\textbf {\bibinfo {volume} {4}},\ \bibinfo
  {pages} {046109} (\bibinfo {year} {2007})}\BibitemShut {NoStop}%
\bibitem [{\citenamefont {Hoogendoorn}\ and\ \citenamefont
  {Bovy}(2004{\natexlab{b}})}]{hoogendoorn-2004b}%
  \BibitemOpen
  \bibfield  {author} {\bibinfo {author} {\bibfnamefont {S.~P.}\ \bibnamefont
  {Hoogendoorn}}\ and\ \bibinfo {author} {\bibfnamefont {P.~H.~L.}\
  \bibnamefont {Bovy}},\ }\href@noop {} {\bibfield  {journal} {\bibinfo
  {journal} {Transportation Research Part B}\ }\textbf {\bibinfo {volume}
  {38}},\ \bibinfo {pages} {571} (\bibinfo {year}
  {2004}{\natexlab{b}})}\BibitemShut {NoStop}%
\bibitem [{\citenamefont {Hartmann}(2010)}]{hartmann-2010}%
  \BibitemOpen
  \bibfield  {author} {\bibinfo {author} {\bibfnamefont {D.}~\bibnamefont
  {Hartmann}},\ }\href {\doibase 10.1088/1367-2630/12/4/043032} {\bibfield
  {journal} {\bibinfo  {journal} {New Journal of Physics}\ }\textbf {\bibinfo
  {volume} {12}},\ \bibinfo {pages} {043032} (\bibinfo {year}
  {2010})}\BibitemShut {NoStop}%
\bibitem [{\citenamefont {Seyfried}\ \emph {et~al.}(2005)\citenamefont
  {Seyfried}, \citenamefont {Steffen}, \citenamefont {Klingsch},\ and\
  \citenamefont {Boltes}}]{seyfried-2005}%
  \BibitemOpen
  \bibfield  {author} {\bibinfo {author} {\bibfnamefont {A.}~\bibnamefont
  {Seyfried}}, \bibinfo {author} {\bibfnamefont {B.}~\bibnamefont {Steffen}},
  \bibinfo {author} {\bibfnamefont {W.}~\bibnamefont {Klingsch}}, \ and\
  \bibinfo {author} {\bibfnamefont {M.}~\bibnamefont {Boltes}},\ }\href
  {http://stacks.iop.org/1742-5468/2005/i=10/a=P10002} {\bibfield  {journal}
  {\bibinfo  {journal} {Journal of Statistical Mechanics: Theory and
  Experiment}\ }\textbf {\bibinfo {volume} {2005}},\ \bibinfo {pages} {P10002}
  (\bibinfo {year} {2005})}\BibitemShut {NoStop}%
\bibitem [{\citenamefont {Chattaraj}\ \emph {et~al.}(2009)\citenamefont
  {Chattaraj}, \citenamefont {Seyfried},\ and\ \citenamefont
  {Chakroborty}}]{chattaraj-2009}%
  \BibitemOpen
  \bibfield  {author} {\bibinfo {author} {\bibfnamefont {U.}~\bibnamefont
  {Chattaraj}}, \bibinfo {author} {\bibfnamefont {A.}~\bibnamefont {Seyfried}},
  \ and\ \bibinfo {author} {\bibfnamefont {P.}~\bibnamefont {Chakroborty}},\
  }\href@noop {} {\bibfield  {journal} {\bibinfo  {journal} {Advances in
  Complex Systems}\ }\textbf {\bibinfo {volume} {12}},\ \bibinfo {pages} {393}
  (\bibinfo {year} {2009})}\BibitemShut {NoStop}%
\bibitem [{\citenamefont {Jeli\'{c}}\ \emph {et~al.}(2012)\citenamefont
  {Jeli\'{c}}, \citenamefont {Appert-Rolland}, \citenamefont {Lemercier},\ and\
  \citenamefont {Pettr\'{e}}}]{jelic-2012}%
  \BibitemOpen
  \bibfield  {author} {\bibinfo {author} {\bibfnamefont {A.}~\bibnamefont
  {Jeli\'{c}}}, \bibinfo {author} {\bibfnamefont {C.}~\bibnamefont
  {Appert-Rolland}}, \bibinfo {author} {\bibfnamefont {S.}~\bibnamefont
  {Lemercier}}, \ and\ \bibinfo {author} {\bibfnamefont {J.}~\bibnamefont
  {Pettr\'{e}}},\ }\href {\doibase 10.1103/PhysRevE.85.036111} {\bibfield
  {journal} {\bibinfo  {journal} {Physical Review E}\ }\textbf {\bibinfo
  {volume} {85}},\ \bibinfo {pages} {036111} (\bibinfo {year}
  {2012})}\BibitemShut {NoStop}%
\bibitem [{\citenamefont {Fiorini}\ and\ \citenamefont
  {Shiller}(1998)}]{fiorini-1998}%
  \BibitemOpen
  \bibfield  {author} {\bibinfo {author} {\bibfnamefont {P.}~\bibnamefont
  {Fiorini}}\ and\ \bibinfo {author} {\bibfnamefont {Z.}~\bibnamefont
  {Shiller}},\ }\href {\doibase 10.1177/027836499801700706} {\bibfield
  {journal} {\bibinfo  {journal} {The International Journal of Robotics
  Research}\ }\textbf {\bibinfo {volume} {17}},\ \bibinfo {pages} {760}
  (\bibinfo {year} {1998})}\BibitemShut {NoStop}%
\bibitem [{\citenamefont {Shiller}\ \emph {et~al.}(2001)\citenamefont
  {Shiller}, \citenamefont {Large},\ and\ \citenamefont
  {Sekhavat}}]{shiller-2001}%
  \BibitemOpen
  \bibfield  {author} {\bibinfo {author} {\bibfnamefont {Z.}~\bibnamefont
  {Shiller}}, \bibinfo {author} {\bibfnamefont {F.}~\bibnamefont {Large}}, \
  and\ \bibinfo {author} {\bibfnamefont {S.}~\bibnamefont {Sekhavat}},\ }in\
  \href {\doibase 10.1109/robot.2001.933196} {\emph {\bibinfo {booktitle} {IEEE
  International Conference on Robotics and Automation, 2001}}}\ (\bibinfo
  {year} {2001})\BibitemShut {NoStop}%
\bibitem [{\citenamefont {Berg}\ \emph {et~al.}(2011)\citenamefont {Berg},
  \citenamefont {Guy}, \citenamefont {Lin},\ and\ \citenamefont
  {Manocha}}]{berg-2011}%
  \BibitemOpen
  \bibfield  {author} {\bibinfo {author} {\bibfnamefont {J.}~\bibnamefont
  {Berg}}, \bibinfo {author} {\bibfnamefont {S.~J.}\ \bibnamefont {Guy}},
  \bibinfo {author} {\bibfnamefont {M.}~\bibnamefont {Lin}}, \ and\ \bibinfo
  {author} {\bibfnamefont {D.}~\bibnamefont {Manocha}},\ }\href {\doibase
  10.1007/978-3-642-19457-3\_1} {\bibfield  {journal} {\bibinfo  {journal}
  {Springer Tracts in Advanced Robotics}\ }\textbf {\bibinfo {volume} {70}},\
  \bibinfo {pages} {3} (\bibinfo {year} {2011})}\BibitemShut {NoStop}%
\bibitem [{\citenamefont {Curtis}\ and\ \citenamefont
  {Manocha}(2014)}]{curtis-2014}%
  \BibitemOpen
  \bibfield  {author} {\bibinfo {author} {\bibfnamefont {S.}~\bibnamefont
  {Curtis}}\ and\ \bibinfo {author} {\bibfnamefont {D.}~\bibnamefont
  {Manocha}},\ }in\ \href {\doibase 10.1007/978-3-319-02447-9_73} {\emph
  {\bibinfo {booktitle} {Pedestrian and Evacuation Dynamics 2012}}},\ \bibinfo
  {editor} {edited by\ \bibinfo {editor} {\bibfnamefont {U.}~\bibnamefont
  {Weidmann}}, \bibinfo {editor} {\bibfnamefont {U.}~\bibnamefont {Kirsch}}, \
  and\ \bibinfo {editor} {\bibfnamefont {M.}~\bibnamefont {Schreckenberg}}}\
  (\bibinfo  {publisher} {Springer International Publishing},\ \bibinfo {year}
  {2014})\ pp.\ \bibinfo {pages} {875--890}\BibitemShut {NoStop}%
\bibitem [{\citenamefont {Moln\'{a}r}(1996)}]{molnar-1996}%
  \BibitemOpen
  \bibfield  {author} {\bibinfo {author} {\bibfnamefont {P.}~\bibnamefont
  {Moln\'{a}r}},\ }\emph {\bibinfo {title} {Modellierung und Simulation der
  Dynamik von Fu\ss{}g\"{a}ngerstr\"{o}men}},\ \href@noop {} {Ph.D. thesis},\
  \bibinfo  {school} {Universit\"{a}t Stuttgart} (\bibinfo {year}
  {1996})\BibitemShut {NoStop}%
\bibitem [{\citenamefont {Moussa\"{i}d}\ \emph {et~al.}(2009)\citenamefont
  {Moussa\"{i}d}, \citenamefont {Helbing}, \citenamefont {Garnier},
  \citenamefont {Johansson}, \citenamefont {Combe},\ and\ \citenamefont
  {Theraulaz}}]{moussaid-2009b}%
  \BibitemOpen
  \bibfield  {author} {\bibinfo {author} {\bibfnamefont {M.}~\bibnamefont
  {Moussa\"{i}d}}, \bibinfo {author} {\bibfnamefont {D.}~\bibnamefont
  {Helbing}}, \bibinfo {author} {\bibfnamefont {S.}~\bibnamefont {Garnier}},
  \bibinfo {author} {\bibfnamefont {A.}~\bibnamefont {Johansson}}, \bibinfo
  {author} {\bibfnamefont {M.}~\bibnamefont {Combe}}, \ and\ \bibinfo {author}
  {\bibfnamefont {G.}~\bibnamefont {Theraulaz}},\ }\href {\doibase
  10.1098/rspb.2009.0405} {\bibfield  {journal} {\bibinfo  {journal}
  {Proceedings of the Royal Society B: Biological Sciences}\ }\textbf {\bibinfo
  {volume} {276}},\ \bibinfo {pages} {2755} (\bibinfo {year}
  {2009})}\BibitemShut {NoStop}%
\bibitem [{\citenamefont {Weidmann}(1992)}]{weidmann-1993}%
  \BibitemOpen
  \bibfield  {author} {\bibinfo {author} {\bibfnamefont {U.}~\bibnamefont
  {Weidmann}},\ }\href {\doibase 10.3929/ethz-a-000687810} {\emph {\bibinfo
  {title} {Transporttechnik der Fussg\"{a}nger}}},\ \bibinfo {edition} {2nd}\
  ed.,\ \bibinfo {series} {Schriftenreihe des IVT}, Vol.~\bibinfo {volume}
  {90}\ (\bibinfo  {publisher} {Institut f\"{u}r Verkehrsplanung,
  Transporttechnik, Strassen- und Eisenbahnbau (IVT) ETH},\ \bibinfo {address}
  {Z\"{u}rich},\ \bibinfo {year} {1992})\BibitemShut {NoStop}%
\bibitem [{\citenamefont {Moussa\"{i}d}\ \emph {et~al.}(2011)\citenamefont
  {Moussa\"{i}d}, \citenamefont {Helbing},\ and\ \citenamefont
  {Theraulaz}}]{moussaid-2011}%
  \BibitemOpen
  \bibfield  {author} {\bibinfo {author} {\bibfnamefont {M.}~\bibnamefont
  {Moussa\"{i}d}}, \bibinfo {author} {\bibfnamefont {D.}~\bibnamefont
  {Helbing}}, \ and\ \bibinfo {author} {\bibfnamefont {G.}~\bibnamefont
  {Theraulaz}},\ }\href {http://www.pnas.org/content/108/17/6884.abstract}
  {\bibfield  {journal} {\bibinfo  {journal} {Proceedings of the National
  Academy of Sciences}\ }\textbf {\bibinfo {volume} {108}},\ \bibinfo {pages}
  {6884} (\bibinfo {year} {2011})}\BibitemShut {NoStop}%
\bibitem [{\citenamefont {Hartmann}\ \emph {et~al.}(2014)\citenamefont
  {Hartmann}, \citenamefont {Mille}, \citenamefont {Pfaffinger},\ and\
  \citenamefont {Royer}}]{hartmann-2014b}%
  \BibitemOpen
  \bibfield  {author} {\bibinfo {author} {\bibfnamefont {D.}~\bibnamefont
  {Hartmann}}, \bibinfo {author} {\bibfnamefont {J.}~\bibnamefont {Mille}},
  \bibinfo {author} {\bibfnamefont {A.}~\bibnamefont {Pfaffinger}}, \ and\
  \bibinfo {author} {\bibfnamefont {C.}~\bibnamefont {Royer}},\ }in\ \href
  {\doibase 10.1007/978-3-319-02447-9_102} {\emph {\bibinfo {booktitle}
  {Pedestrian and Evacuation Dynamics 2012}}},\ \bibinfo {editor} {edited by\
  \bibinfo {editor} {\bibfnamefont {U.}~\bibnamefont {Weidmann}}, \bibinfo
  {editor} {\bibfnamefont {U.}~\bibnamefont {Kirsch}}, \ and\ \bibinfo {editor}
  {\bibfnamefont {M.}~\bibnamefont {Schreckenberg}}}\ (\bibinfo  {publisher}
  {Springer International Publishing},\ \bibinfo {year} {2014})\ pp.\ \bibinfo
  {pages} {1237--1249}\BibitemShut {NoStop}%
\bibitem [{\citenamefont {Treuille}\ \emph {et~al.}(2006)\citenamefont
  {Treuille}, \citenamefont {Cooper},\ and\ \citenamefont
  {Popovi\'{c}}}]{treuille-2006}%
  \BibitemOpen
  \bibfield  {author} {\bibinfo {author} {\bibfnamefont {A.}~\bibnamefont
  {Treuille}}, \bibinfo {author} {\bibfnamefont {S.}~\bibnamefont {Cooper}}, \
  and\ \bibinfo {author} {\bibfnamefont {Z.}~\bibnamefont {Popovi\'{c}}},\
  }\href {\doibase 10.1145/1141911.1142008} {\bibfield  {journal} {\bibinfo
  {journal} {ACM Transactions on Graphics (SIGGRAPH 2006)}\ }\textbf {\bibinfo
  {volume} {25}},\ \bibinfo {pages} {1160} (\bibinfo {year}
  {2006})}\BibitemShut {NoStop}%
\bibitem [{\citenamefont {Evans}(1997)}]{evans-1997}%
  \BibitemOpen
  \bibfield  {author} {\bibinfo {author} {\bibfnamefont {L.~C.}\ \bibnamefont
  {Evans}},\ }\href@noop {} {\emph {\bibinfo {title} {Partial Differential
  Equations}}}\ (\bibinfo  {publisher} {American Mathematical Society},\
  \bibinfo {year} {1997})\ p.\ \bibinfo {pages} {664}\BibitemShut {NoStop}%
\bibitem [{\citenamefont {Johansson}\ \emph {et~al.}(2007)\citenamefont
  {Johansson}, \citenamefont {Helbing},\ and\ \citenamefont
  {Shukla}}]{johansson-2007}%
  \BibitemOpen
  \bibfield  {author} {\bibinfo {author} {\bibfnamefont {A.}~\bibnamefont
  {Johansson}}, \bibinfo {author} {\bibfnamefont {D.}~\bibnamefont {Helbing}},
  \ and\ \bibinfo {author} {\bibfnamefont {P.}~\bibnamefont {Shukla}},\
  }\href@noop {} {\bibfield  {journal} {\bibinfo  {journal} {Advances in
  Complex Systems}\ }\textbf {\bibinfo {volume} {10}},\ \bibinfo {pages} {271}
  (\bibinfo {year} {2007})}\BibitemShut {NoStop}%
\bibitem [{\citenamefont {Sethian}(1999)}]{sethian-1999}%
  \BibitemOpen
  \bibfield  {author} {\bibinfo {author} {\bibfnamefont {J.~A.}\ \bibnamefont
  {Sethian}},\ }\href@noop {} {\emph {\bibinfo {title} {Level Set Methods and
  Fast Marching Methods: Evolving Interfaces in Computational Geometry, Fluid
  Mechanics, Computer Vision, and Materials Science}}}\ (\bibinfo  {publisher}
  {Cambridge University Press},\ \bibinfo {year} {1999})\BibitemShut {NoStop}%
\bibitem [{\citenamefont {Kretz}\ \emph {et~al.}(2006)\citenamefont {Kretz},
  \citenamefont {Gr\"{u}nebohm},\ and\ \citenamefont
  {Schreckenberg}}]{kretz-2006}%
  \BibitemOpen
  \bibfield  {author} {\bibinfo {author} {\bibfnamefont {T.}~\bibnamefont
  {Kretz}}, \bibinfo {author} {\bibfnamefont {A.}~\bibnamefont
  {Gr\"{u}nebohm}}, \ and\ \bibinfo {author} {\bibfnamefont {M.}~\bibnamefont
  {Schreckenberg}},\ }\href {\doibase 10.1088/1742-5468/2006/10/P10014}
  {\bibfield  {journal} {\bibinfo  {journal} {Journal of Statistical Mechanics:
  Theory and Experiment}\ }\textbf {\bibinfo {volume} {2006}},\ \bibinfo
  {pages} {P10014} (\bibinfo {year} {2006})}\BibitemShut {NoStop}%
\bibitem [{\citenamefont {Seyfried}\ \emph
  {et~al.}(2009{\natexlab{a}})\citenamefont {Seyfried}, \citenamefont {Passon},
  \citenamefont {Steffen}, \citenamefont {Boltes}, \citenamefont {Rupprecht},\
  and\ \citenamefont {Klingsch}}]{seyfried-2009}%
  \BibitemOpen
  \bibfield  {author} {\bibinfo {author} {\bibfnamefont {A.}~\bibnamefont
  {Seyfried}}, \bibinfo {author} {\bibfnamefont {O.}~\bibnamefont {Passon}},
  \bibinfo {author} {\bibfnamefont {B.}~\bibnamefont {Steffen}}, \bibinfo
  {author} {\bibfnamefont {M.}~\bibnamefont {Boltes}}, \bibinfo {author}
  {\bibfnamefont {T.}~\bibnamefont {Rupprecht}}, \ and\ \bibinfo {author}
  {\bibfnamefont {W.}~\bibnamefont {Klingsch}},\ }\href@noop {} {\bibfield
  {journal} {\bibinfo  {journal} {Transportation Science}\ }\textbf {\bibinfo
  {volume} {43}},\ \bibinfo {pages} {395} (\bibinfo {year}
  {2009}{\natexlab{a}})}\BibitemShut {NoStop}%
\bibitem [{\citenamefont {Liddle}\ \emph {et~al.}(2011)\citenamefont {Liddle},
  \citenamefont {Seyfried}, \citenamefont {Steffen}, \citenamefont {Klingsch},
  \citenamefont {Rupprecht}, \citenamefont {Winkens},\ and\ \citenamefont
  {Boltes}}]{liddle-2011}%
  \BibitemOpen
  \bibfield  {author} {\bibinfo {author} {\bibfnamefont {J.}~\bibnamefont
  {Liddle}}, \bibinfo {author} {\bibfnamefont {A.}~\bibnamefont {Seyfried}},
  \bibinfo {author} {\bibfnamefont {B.}~\bibnamefont {Steffen}}, \bibinfo
  {author} {\bibfnamefont {W.}~\bibnamefont {Klingsch}}, \bibinfo {author}
  {\bibfnamefont {T.}~\bibnamefont {Rupprecht}}, \bibinfo {author}
  {\bibfnamefont {A.}~\bibnamefont {Winkens}}, \ and\ \bibinfo {author}
  {\bibfnamefont {M.}~\bibnamefont {Boltes}},\ }\href@noop {} {\bibfield
  {journal} {\bibinfo  {journal} {arXiv}\ }\textbf {\bibinfo {volume}
  {1105.1532}},\ \bibinfo {pages} {v1} (\bibinfo {year} {2011})}\BibitemShut
  {NoStop}%
\bibitem [{\citenamefont {Seyfried}\ \emph
  {et~al.}(2009{\natexlab{b}})\citenamefont {Seyfried}, \citenamefont
  {Steffen}, \citenamefont {Winkens}, \citenamefont {Rupprecht}, \citenamefont
  {Boltes},\ and\ \citenamefont {Klingsch}}]{seyfried-2009b}%
  \BibitemOpen
  \bibfield  {author} {\bibinfo {author} {\bibfnamefont {A.}~\bibnamefont
  {Seyfried}}, \bibinfo {author} {\bibfnamefont {B.}~\bibnamefont {Steffen}},
  \bibinfo {author} {\bibfnamefont {A.}~\bibnamefont {Winkens}}, \bibinfo
  {author} {\bibfnamefont {T.}~\bibnamefont {Rupprecht}}, \bibinfo {author}
  {\bibfnamefont {M.}~\bibnamefont {Boltes}}, \ and\ \bibinfo {author}
  {\bibfnamefont {W.}~\bibnamefont {Klingsch}},\ }in\ \href {\doibase
  10.1007/978-3-540-77074-9\_17} {\emph {\bibinfo {booktitle} {Traffic and
  Granular Flow 2007}}},\ \bibinfo {editor} {edited by\ \bibinfo {editor}
  {\bibfnamefont {C.}~\bibnamefont {Appert-Rolland}}, \bibinfo {editor}
  {\bibfnamefont {F.}~\bibnamefont {Chevoir}}, \bibinfo {editor} {\bibfnamefont
  {P.}~\bibnamefont {Gondret}}, \bibinfo {editor} {\bibfnamefont
  {S.}~\bibnamefont {Lassarre}}, \bibinfo {editor} {\bibfnamefont {J.-P.}\
  \bibnamefont {Lebacque}}, \ and\ \bibinfo {editor} {\bibfnamefont
  {M.}~\bibnamefont {Schreckenberg}}}\ (\bibinfo  {publisher} {Springer Berlin
  Heidelberg},\ \bibinfo {year} {2009})\ pp.\ \bibinfo {pages}
  {189--199}\BibitemShut {NoStop}%
\bibitem [{\citenamefont {Schadschneider}\ and\ \citenamefont
  {Seyfried}(2011)}]{schadschneider-2011b}%
  \BibitemOpen
  \bibfield  {author} {\bibinfo {author} {\bibfnamefont {A.}~\bibnamefont
  {Schadschneider}}\ and\ \bibinfo {author} {\bibfnamefont {A.}~\bibnamefont
  {Seyfried}},\ }\href
  {http://aimsciences.org/journals/displayArticlesnew.jsp?paperID=6445}
  {\bibfield  {journal} {\bibinfo  {journal} {Networks and Heterogeneous
  Media}\ }\textbf {\bibinfo {volume} {6}},\ \bibinfo {pages} {545} (\bibinfo
  {year} {2011})}\BibitemShut {NoStop}%
\bibitem [{\citenamefont {Helbing}\ \emph {et~al.}(2000)\citenamefont
  {Helbing}, \citenamefont {Farkas},\ and\ \citenamefont
  {Vicsek}}]{helbing-2000}%
  \BibitemOpen
  \bibfield  {author} {\bibinfo {author} {\bibfnamefont {D.}~\bibnamefont
  {Helbing}}, \bibinfo {author} {\bibfnamefont {I.}~\bibnamefont {Farkas}}, \
  and\ \bibinfo {author} {\bibfnamefont {T.}~\bibnamefont {Vicsek}},\
  }\href@noop {} {\bibfield  {journal} {\bibinfo  {journal} {Nature}\ }\textbf
  {\bibinfo {volume} {407}},\ \bibinfo {pages} {487} (\bibinfo {year}
  {2000})}\BibitemShut {NoStop}%
\bibitem [{\citenamefont {Gerhardt}\ \emph {et~al.}(2011)\citenamefont
  {Gerhardt}, \citenamefont {K\"{o}ster}, \citenamefont {Seitz}, \citenamefont
  {Treml},\ and\ \citenamefont {Klein}}]{gerhardt-2011}%
  \BibitemOpen
  \bibfield  {author} {\bibinfo {author} {\bibfnamefont {K.}~\bibnamefont
  {Gerhardt}}, \bibinfo {author} {\bibfnamefont {G.}~\bibnamefont
  {K\"{o}ster}}, \bibinfo {author} {\bibfnamefont {M.}~\bibnamefont {Seitz}},
  \bibinfo {author} {\bibfnamefont {F.}~\bibnamefont {Treml}}, \ and\ \bibinfo
  {author} {\bibfnamefont {W.}~\bibnamefont {Klein}},\ }in\ \href@noop {}
  {\emph {\bibinfo {booktitle} {Proceedings of the International Conference on
  Emergency Evacuation of People from Buildings}}}\ (\bibinfo {address}
  {Warzaw, Poland},\ \bibinfo {year} {2011})\BibitemShut {NoStop}%
\bibitem [{\citenamefont {Gaididei}\ \emph {et~al.}(2013)\citenamefont
  {Gaididei}, \citenamefont {Gorria}, \citenamefont {Berkemer}, \citenamefont
  {Kawamoto}, \citenamefont {Shiga}, \citenamefont {Christiansen},
  \citenamefont {S{\o}rensen},\ and\ \citenamefont {Starke}}]{gaididei-2013}%
  \BibitemOpen
  \bibfield  {author} {\bibinfo {author} {\bibfnamefont {Y.~B.}\ \bibnamefont
  {Gaididei}}, \bibinfo {author} {\bibfnamefont {C.}~\bibnamefont {Gorria}},
  \bibinfo {author} {\bibfnamefont {R.}~\bibnamefont {Berkemer}}, \bibinfo
  {author} {\bibfnamefont {A.}~\bibnamefont {Kawamoto}}, \bibinfo {author}
  {\bibfnamefont {T.}~\bibnamefont {Shiga}}, \bibinfo {author} {\bibfnamefont
  {P.~L.}\ \bibnamefont {Christiansen}}, \bibinfo {author} {\bibfnamefont
  {M.~P.}\ \bibnamefont {S{\o}rensen}}, \ and\ \bibinfo {author} {\bibfnamefont
  {J.}~\bibnamefont {Starke}},\ }\href {\doibase 10.1103/physreve.88.042803}
  {\bibfield  {journal} {\bibinfo  {journal} {Physical Review E}\ }\textbf
  {\bibinfo {volume} {88}},\ \bibinfo {pages} {042803} (\bibinfo {year}
  {2013})}\BibitemShut {NoStop}%
\bibitem [{\citenamefont {Marschler}\ \emph {et~al.}(2013)\citenamefont
  {Marschler}, \citenamefont {Sieber}, \citenamefont {Berkemer}, \citenamefont
  {Kawamoto},\ and\ \citenamefont {Starke}}]{marschler-2013}%
  \BibitemOpen
  \bibfield  {author} {\bibinfo {author} {\bibfnamefont {C.}~\bibnamefont
  {Marschler}}, \bibinfo {author} {\bibfnamefont {J.}~\bibnamefont {Sieber}},
  \bibinfo {author} {\bibfnamefont {R.}~\bibnamefont {Berkemer}}, \bibinfo
  {author} {\bibfnamefont {A.}~\bibnamefont {Kawamoto}}, \ and\ \bibinfo
  {author} {\bibfnamefont {J.}~\bibnamefont {Starke}},\ }\href
  {http://arxiv.org/abs/1301.6044} {\bibfield  {journal} {\bibinfo  {journal}
  {arXiv}\ }\textbf {\bibinfo {volume} {1301.6044}},\ \bibinfo {pages} {v1}
  (\bibinfo {year} {2013})}\BibitemShut {NoStop}%
\bibitem [{\citenamefont {Portz}\ and\ \citenamefont
  {Seyfried}(2011)}]{portz-2011}%
  \BibitemOpen
  \bibfield  {author} {\bibinfo {author} {\bibfnamefont {A.}~\bibnamefont
  {Portz}}\ and\ \bibinfo {author} {\bibfnamefont {A.}~\bibnamefont
  {Seyfried}},\ }\href {\doibase 10.1007/978-1-4419-9725-8_52} {\bibfield
  {journal} {\bibinfo  {journal} {Pedestrian and Evacuation Dynamics}\ }\textbf
  {\bibinfo {volume} {1}},\ \bibinfo {pages} {577} (\bibinfo {year}
  {2011})}\BibitemShut {NoStop}%
\bibitem [{\citenamefont {Zhang}\ \emph {et~al.}(2012)\citenamefont {Zhang},
  \citenamefont {Klingsch}, \citenamefont {Schadschneider},\ and\ \citenamefont
  {Seyfried}}]{zhang-2012b}%
  \BibitemOpen
  \bibfield  {author} {\bibinfo {author} {\bibfnamefont {J.}~\bibnamefont
  {Zhang}}, \bibinfo {author} {\bibfnamefont {W.}~\bibnamefont {Klingsch}},
  \bibinfo {author} {\bibfnamefont {A.}~\bibnamefont {Schadschneider}}, \ and\
  \bibinfo {author} {\bibfnamefont {A.}~\bibnamefont {Seyfried}},\ }\href
  {http://stacks.iop.org/1742-5468/2012/i=02/a=P02002} {\bibfield  {journal}
  {\bibinfo  {journal} {Journal of Statistical Mechanics: Theory and
  Experiment}\ }\textbf {\bibinfo {volume} {2012}},\ \bibinfo {pages} {P02002}
  (\bibinfo {year} {2012})}\BibitemShut {NoStop}%
\end{thebibliography}%

\appendix*

\section*{Appendix A}

\subsection{\label{sec:smoothg}Vector normalizer}
In order to design a function that smoothly scales a given vector to a length in $[0,1]$, a smooth ramp function is needed. The following chain of definitions is adopted from \cite{koster-2013}:
Let $r:\real\to\real$ be the ramp function defined by
\begin{equation}
r(x)=\begin{cases}
0&\text{for }x\leq 0\\
x&\text{for }x\in(0,1)\\
1&\text{for }x\geq 1\\
\end{cases}
\end{equation}
Then, a smooth version $r_\text{moll,p}$ with mollification parameter $p$ is given by
\begin{eqnarray}
\text{moll}(x,R,p)&=&\begin{cases}
e\cdot exp(\frac{1}{(\|x\|/R)^{2p}-1})&\text{for }\|x\|<R\\
0&\text{for }\|x\|\geq R
\end{cases}\\
r_\text{moll,p}(x)&=&\text{moll}(x,1,p)\cdot x +(1-\text{moll}(x,1,p))
\end{eqnarray}
where $p>1$. For this paper, we used $p=3$.
\begin{lemma}
The following two statements hold:
\begin{enumerate}
\item[(i)] $r_\text{moll,p}\in C^\infty(\real)$
\item[(ii)] $r_\text{moll,p}(x)=r(x)\ \forall x\in \real\setminus (0,1)$
\end{enumerate}
\end{lemma}
\begin{proof}
(i) holds since the standard mollifier is smooth: $\text{moll}(x,R,p)\in C^\infty$ \cite{evans-1997}. (ii) is trivial from the definitions of $r$ and $r_{\text{moll,p}}$.
\end{proof}
The desired scaling function $g$ can now be defined as follows:
\begin{eqnarray*}
g:\real^2&\to&\real^2\\
x&\mapsto&\begin{cases}
(0,0)^T&\text{for }\|x\|=0\\
x/\|x\|\cdot r_{moll,p}(\|x\|)&\text{for }\|x\|>0
\end{cases}
\end{eqnarray*}
For a similar, one-dimensional version, for example for smoothing $max(0,x)$ with $x\in[-1,1]$, the logistic function can be used:
\begin{equation}\label{eq:tildeg}
\tilde{g}(x;x_0,R)=\frac{1}{1 + e^{-(x-x_0) / R}}
\end{equation}
with $x,x_0,R\in\mathbb{R}$. In this paper, we choose $x_0=0.3$ and $R=0.03$ to smooth the influence of the viewing angle.

\end{document}